\documentclass[aps,amsmath,amssymb,twocolumn,prl,superscriptaddress]{revtex4-2}

\usepackage{graphicx,tikz}
\usepackage{dcolumn}
\usepackage{bm}
\usepackage{epsfig,hyperref,amsthm}
\usepackage{mathtools}
\usepackage{color,fancybox}
\usepackage{colordvi}
\usepackage{relsize}
\usepackage{accents}
\usepackage{wasysym} 
\usepackage{enumitem}
\usepackage{mathtools,slashed}
\usepackage{times}
\usepackage{newtxmath}
\usepackage[displaymath, mathlines]{lineno}
\usepackage{soul}

\hypersetup{
	colorlinks=true,
	linkcolor=CitingColor,
	citecolor=CitingColor,
	urlcolor=CitingColor
}

\DeclareMathAlphabet\mathbfcal{OMS}{cmsy}{b}{n}
\newcommand{\ket}[1]{\ensuremath{|#1\rangle}}
\newcommand{\bra}[1]{\ensuremath{\langle #1|}}
\newcommand{\braket}[2]{\langle #1|#2\rangle}
\newcommand{\proj}[1]{\ket{#1}\bra{#1}}

\newcommand{\tr}{{\rm tr}}
\newcommand{\norm}[1]{\left\|#1\right\|}

\newcommand{\id}{\mathbb{I}}

\newtheorem{theorem}{Theorem}

\newtheorem{lemma}[theorem]{Lemma}

\newtheorem{result}{Result}

\newtheorem{question}{Question}
\newtheorem{fact}[theorem]{Fact}

\definecolor{nred}{rgb}{0.9,0.1,0.1}
\definecolor{nblack}{rgb}{0,0,0}
\definecolor{nblue}{rgb}{0.2,0.2,0.8}
\definecolor{ngreen}{rgb}{0.2,0.6,0.2}
\definecolor{ublue}{rgb}{0,0,0.5}

\definecolor{pur}{rgb}{0.75,0,0.75}
\definecolor{nngrn}{rgb}{0,0.5,0.5}
\definecolor{CitingColor}{rgb}{0,0.3,1}

\newcommand{\blu}{\color{nblue}}

\newcommand{\CY}[1]{{\color{black}#1}}
\newcommand{\CYtwo}[1]{{\color{black}#1}}
\newcommand{\CYthree}[1]{{\color{black}#1}}

\newcommand{\Sin}{{S_{\rm in}}}
\newcommand{\Sout}{{S_{\rm out}}}
\newcommand{\dSin}{{d_{\rm in}}}
\newcommand{\dSout}{{d_{\rm out}}}
\newcommand{\Choi}{{J}}

\begin{document}
\title{
No-go theorems on simulating uncertainty principle’s signatures
}

\author{Chung-Yun Hsieh}
\affiliation{H. H. Wills Physics Laboratory, University of Bristol, Tyndall Avenue, Bristol, BS8 1TL, United Kingdom}

\author {Minjeong Song}
\email{song.at.qit@gmail.com}
\affiliation{Centre for Quantum Technologies, National University of Singapore, 3 Science Drive 2, Singapore 117543}

\author{Shin-Liang Chen}
\email{shin.liang.chen@email.nchu.edu.tw}
\affiliation{Department of Physics, National Chung Hsing University, Taichung 40227, Taiwan}
\affiliation{Physics Division, National Center for Theoretical Sciences, Taipei 106319, Taiwan}

\date{\today}

\begin{abstract}
Uncertainty principle, one of the most iconic features of quantum mechanics, was originally viewed as a fundamental limitation.
Since the inception of quantum information science, researchers began to use it to achieve quantum advantages. 
To better understand the origin of these advantages, an essential question is: {\em To what extent can the uncertainty principle's signatures be simulated by a single measurement?} 
As a single measurement clearly cannot demonstrate the uncertainty principle, such a simulation, if exists, implies the claimed advantages may either stem from other quantum features, or just be reproducible in a less resourceful way.
In this work, we report a series of noise-robust no-go theorems, showing that strong enough signatures of uncertainty principle {\em cannot} be simulated by a single measurement, even when assisted by quantum pre- or post-processing.
This signature is modelled by {\em complementary instruments}.
We completely characterise complementary instruments by a numerically feasible measure and show that they are necessary and sufficient resources for the advantage in an operational task that aims to unambiguously send classical information. 
\end{abstract}

\maketitle

\section{Introduction}
Uncertainty principle, one of the most profound no-go theorems in quantum theory, sets the impossibility to simultaneously measure two non-commuting physical observables with arbitrary precision.
Initially, it was viewed as a fundamental limitation on what we can ever achieve. 
During the development of quantum information science, researchers started to realise that uncertainty principle's effects can actually act as {\em resources} for different operational tasks.
Nowadays, different notions of uncertainty principle's signatures~\cite{Guhne2023RMP,Busch2006,Busch2007,Kiukas2019,Saha2020PRA,Huang2024PRA,Hsieh-IP,Rolino2025arXiv,Starke2026review,Heinosaari2014,Heinosaari2016,Buscemi2023Quantum,Ghai2025,Ji2023,Hsieh2024PRL,Ghai2025} have been found as resources in, e.g., violating Bell inequalities~\cite{Brunner2014RMP,Bell1964}, demonstrating quantum steering~\cite{XiangPRXQ2022,Cavalcanti2016,Uola2020RMP,Uola2015PRL} as well as different types of device-independent tasks~\cite{Ku2022NC,Ku2022PRXQ,Ku2023,Hsieh2023}, discrimination tasks~\cite{Ji2023,Skrzypczyk2019PRL,Takagi2019,Uola2019PRL,Hsieh2025PRA-3}, exclusion tasks~\cite{Saha2020PRA,Huang2024PRA,Hsieh-IP,Ducuara2020PRL,Uola2020PRL,Carmeli2019PRL,Oszmaniec2019Quantum}, and thermodynamic work extraction~\cite{Hsieh2024PRL,Beyer2019PRL,Ji2022PRL,Biswas2025PRL,Hsieh2025arXiv}, to name a few.  
Hence, rather than a limitation, uncertainty principle enables us to achieve what we could never achieve without its help.

Ultimately, uncertainty principle is a quantum feature that can {\em only} occur when there are more than {\em two} measurement processes.
An important question is: 
\begin{center}
{\em When can we simulate it by a single measurement device?}
\end{center}
Foundationally, this is a crucial problem, as its answer can uncover properties shared by two measurement processes that are responsible for demonstrating uncertainty principle. 
Also, it can tell us when we can really view two measurement processes as genuinely distinct.
Moreover, as uncertainty principle's signatures are powerful resources, it is crucial to know whether this resource can be achieved in a ``cheaper'' way; namely, by a setting involving only a single measurement process.

\begin{figure}[h]
\scalebox{0.725}{\includegraphics{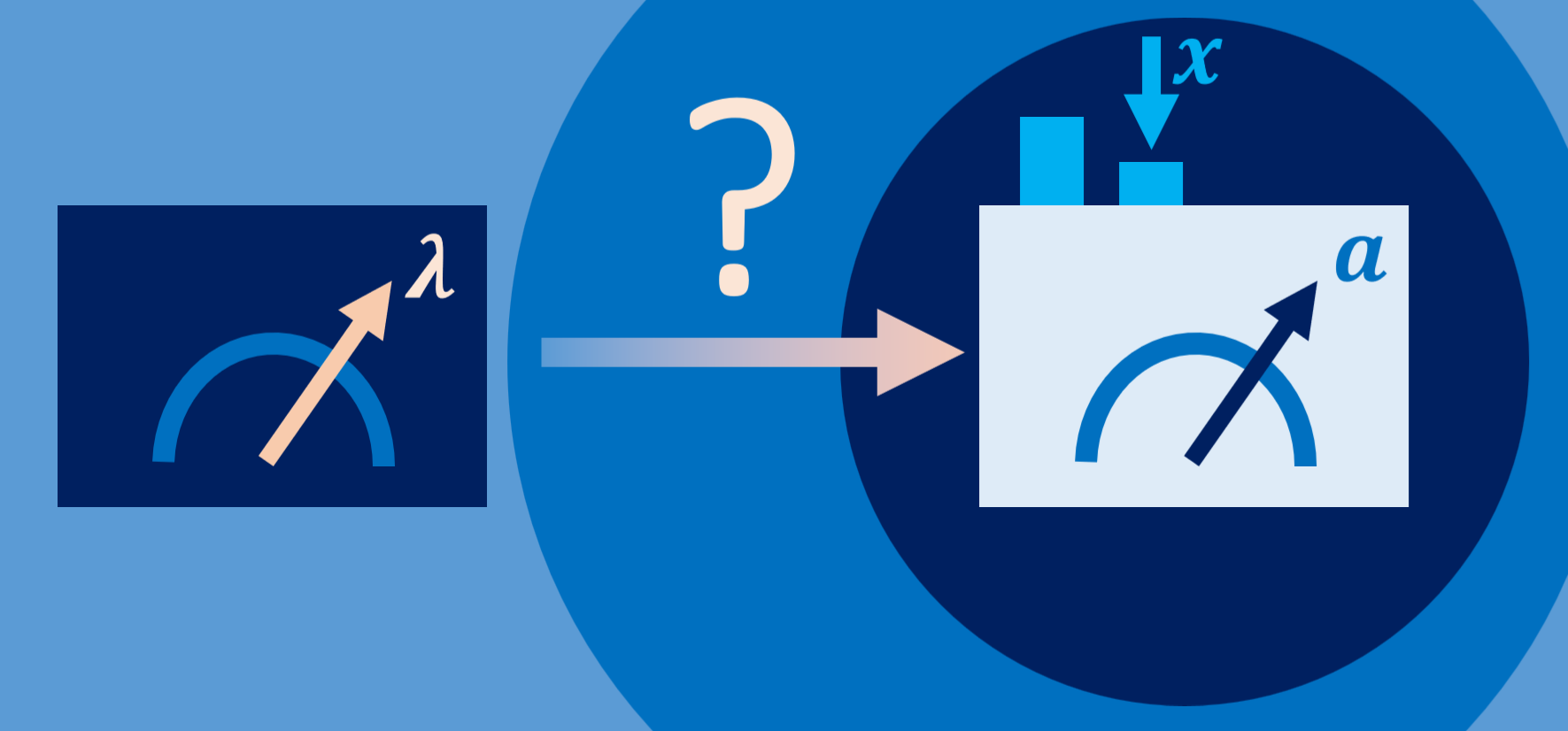}}
\caption{
{\bf Main question.}
When can uncertainty principle's signatures, as carried by a pair of instruments, be simulated by a single instrument? 
In the figure, the left box is an instrument with classical outcome $\lambda$, and the right box implements the $x$-th instrument (with classical outcomes $a$) when we press the $x$-th button.
}
\label{Fig}
\end{figure}

In this work, we report a series of general no-go results on the impossibility of simulating the uncertainty principle's signatures by a single measurement device 
(see also Fig.~\ref{Fig}).
Using the mathematical notion termed {\em instruments}~\cite{QMeasurement-book} to model measurement processes, we consider two operational ways of characterising uncertainty principle, termed {\em incompatible instruments} and {\em complementary instruments}.
The former describes the impossibility of simulating two instruments by {\em joint measurability}~\footnote{Incompatibility considered in this work is also known as {\em traditional incompatibility} in the literature (see, e.g., Ref.~\cite{Buscemi2023Quantum} for a brief review). This is different from the {\em parallel incompatibility}~\cite{Mitra2022PRA,Mitra2023PRA,Leppajarvi2024Quantum}, which is a notion more relevant to the no-cloning theorem.}, which refers to a single instrument with classical post-processing~\cite{Heinosaari2014,Heinosaari2016,Buscemi2023Quantum,Ghai2025,Ji2023,Hsieh2024PRL}. 
The latter is formalised in this work, which extends the operational notions of complementarity~\cite{Busch2006,Busch2007,Kiukas2019,Saha2020PRA,Huang2024PRA,Hsieh-IP,Guhne2023RMP,Rolino2025arXiv,Starke2026review} to instruments, describing the impossibility for two instruments' outputs to overlap.
These seemingly distinct notions are actually closely related.
We show that whether instruments are complementary can be fully determined by a semi-definite program (SDP)~\cite{watrous_2018,SDP-textbook} (Result~\ref{Result:Characterising complementarity}).
This further implies that complementary instruments must also be incompatible (Result~\ref{Result:relation}), meaning that joint measurability cannot simulate them.
We then proceed to explore more advanced simulation models beyond joint measurability and report a series of noise-robust no-go results. 
We show that all rank-one projective complementary instruments, even being noisy, still cannot be simulated by joint measurability assisted by quantum pre- or post-processing (Results~\ref{Result:no-go}-\ref{Result:no-go2}).
Hence, strong enough uncertainty principle's signatures can be pretty hard for a single instrument to simulate. 
Finally, we consider an unambiguous classical communication task for instruments and show that it is possible only via complementary instruments (Result~\ref{Result:operational meaning}).
This thus operationally separates complementarity and incompatibility, uncovering complementarity's advantage in sending classical information.

\section{Results}
\subsection{Channel, filters, and instruments}

To start with, {\em channels} are completely-positive trace-preserving linear maps~\cite{QIC-book}.
A channel $\mathcal{E}$ describes quantum state's deterministic evolution: It maps an initial state $\rho$ to a final state $\mathcal{E}(\rho)$.
A probabilistic evolution is described by a {\em filter}, i.e., a completely-positive trace-non-increasing linear map.
A filter $\mathcal{L}$ maps an initial state $\rho$ to the final state $\mathcal{L}(\rho)/{\rm tr}[\mathcal{L}(\rho)]$ with the probability ${\rm tr}[\mathcal{L}(\rho)]$.
Now, an {\em instrument} is a set $\{\mathcal{E}_a\}_a$ such that each $\mathcal{E}_a$ is a filter, and $\sum_a\mathcal{E}_a$ is a channel~\cite{QMeasurement-book}.
It describes the process that maps an initial state $\rho$ to the $a$-th final state $\mathcal{E}_a(\rho)/{\rm tr}[\mathcal{E}_a(\rho)]$ with the probability ${\rm tr}[\mathcal{E}_a(\rho)]$.
From here, one can observe that instruments naturally describe general measurement processes by including both the classical outcomes ``$a$'' and the updated states ``$\mathcal{E}_a(\rho)/{\rm tr}[\mathcal{E}_a(\rho)]$''.
They have been lately applied to study quantum features of measurement processes~\cite{Hsieh2026PRA,Hsieh2024PRL,Ji2023,Cobucci2026,Khandelwal2025PRL,Buscemi2023Quantum,Ghai2025,Mitra2022PRA,Mitra2023PRA,Leppajarvi2024Quantum} beyond {\em positive operator-valued measures} (POVMs)~\footnote{Formally, a POVM is a set of operators $\{E_a\}_a$ such that $E_a\ge0$ $\forall\,a$ and $\sum_aE_a=\id$~\cite{QIC-book}. 
For an initial state $\rho$, it produces the classical outcome $a$ with probability ${\rm tr}(\rho E_a)$.}. 
Here, we aim to utilise instruments to study uncertainty principle's signatures.

\subsection{Uncertainty principle's signature from incompatible instruments}
Since now, we use $\mathbfcal{E}=\{\mathcal{E}_{a|x}\}_{a,x}$ to denote a set of instruments, i.e., for each $x$, the set $\{\mathcal{E}_{a|x}\}_{a}$ is an instrument. 
Then $\mathbfcal{E}=\{\mathcal{E}_{a|x}\}_{a,x}$ is {\em jointly measurable} (see, e.g., Refs.~\cite{Heinosaari2014,Heinosaari2016,Buscemi2023Quantum,Ghai2025,Ji2023,Hsieh2024PRL,Mitra2022PRA,Mitra2023PRA,Ghai2025,Leppajarvi2024Quantum})
if there exist conditional probabilities $\{P(a|x,\lambda)\}_{a,x,\lambda}$ and a single instrument $\{\mathcal{G}_\lambda\}_\lambda$ such that
\begin{align}\label{Eq:JM}
\mathcal{E}_{a|x} = \sum_\lambda P(a|x,\lambda)\mathcal{G}_\lambda\quad\forall\,a,x.
\end{align}
Namely, $\mathbfcal{E}$ can be simulated by a single instrument assisted by classical post-processing (denoted as $\mathbfcal{E}\in{\bf JM}$).
Measurement processes modelled by jointly-measurable instruments can be measured simultaneously via a single measurement process ($\{\mathcal{G}_\lambda\}_\lambda$), meaning that there is only one genuinely quantum measurement.
We then say $\mathbfcal{E}$ is {\em incompatible} if it is not jointly measurable (i.e., $\mathbfcal{E}\notin{\bf JM}$).
\CYthree{Physically, this means that} $\mathbfcal{E}$ must include at least two genuinely different quantum measurement processes that cannot be implemented simultaneously.
In this sense, incompatibility has been viewed as a signature of uncertainty principle~\cite{Guhne2023RMP}. See also \CY{Appendix A} for discussions.

Note that if two instruments in $\mathbfcal{E}$ are already incompatible, then $\mathbfcal{E}$ is automatically incompatible.
Hence, since now, we focus on a pair of instruments (i.e., $x=0,1$), as the results should also apply to a set with more instruments.

\subsection{Uncertainty principle's signature from complementary instruments}
There is an alternative form of uncertainty principle's signature,
termed {\em complementary instruments}, as we introduce now.
Conceptually, it describes when two instruments' quantum outputs {\em cannot overlap,} suggesting these two instruments cannot be known simultaneously.
To operationally define this, we follow the approach of Ref.~\cite{Hsieh-IP}.
For a pair of instruments \CYtwo{$\mathbfcal{E} = \{\mathcal{E}_{a|x}\}_{a,x}$ with $x=0,1$,}
we say it has \CY{\em non-vanishing overlap} if there exist \CYtwo{output index pair $(a,b)$
and a \CY{non-vanishing} filter $\mathcal{L}\neq0$ 
such that (see also Fig.~\ref{Fig:overlap})}
\begin{align}\label{Eq:common part 01}
&\CYtwo{\mathcal{E}_{a|0}-\mathcal{L}\;\;\&\;\;\mathcal{E}_{b|1}-\mathcal{L}\;\text{are both completely-positive.}}
\end{align}
Namely, $\mathcal{E}_{a|0}$ and $\mathcal{E}_{b|1}$ have certain information (via measurement statistics) that can be simultaneously obtained via $\mathcal{L}$ (the ``overlap'').
We then say $\mathbfcal{E}$ is {\em complementary} if it has {\em vanishing} overlap, \CYtwo{meaning that
no outcome pair carries ``common parts'' described by a single filter.
This thus gives us a novel way to identify uncertainty principle's signature.}

\begin{figure}[h]
\scalebox{0.725}{\includegraphics{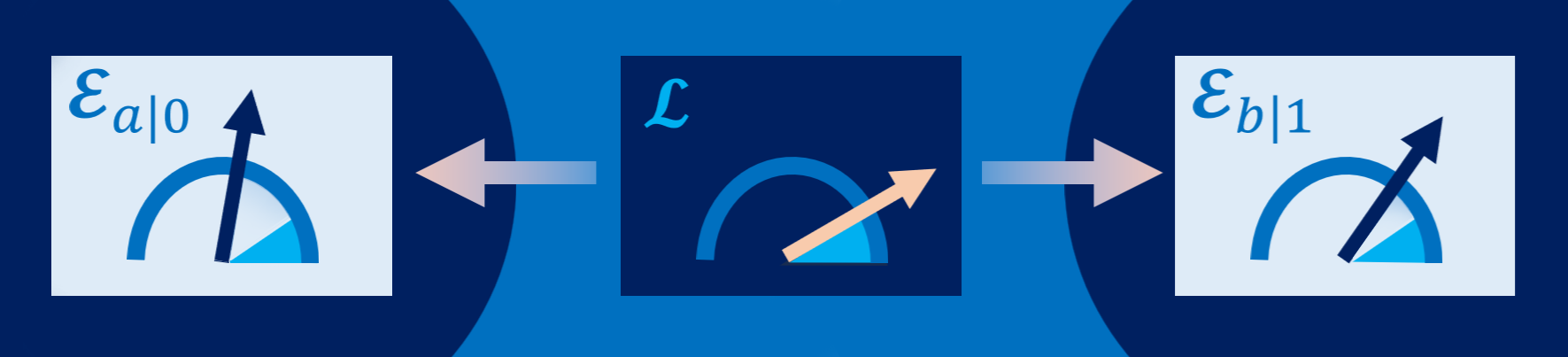}}
\caption{
{\bf Non-vanishing overlap.} \CYtwo{Equation~\eqref{Eq:common part 01} implies that \mbox{$(\mathcal{E}_{a|0}\otimes\mathcal{I}_A)(\rho)\ge(\mathcal{L}\otimes\mathcal{I}_A)(\rho)$} and $(\mathcal{E}_{b|1}\otimes\mathcal{I}_A)(\rho)\ge(\mathcal{L}\otimes\mathcal{I}_A)(\rho)$ for every finite-dimensional auxiliary system $A$ and joint state $\rho$ (here, $\mathcal{I}_A$ is the identity channel acting on $A$). 
Hence, both the classical outcomes and quantum outputs of $\mathcal{E}_{a|0}$ and $\mathcal{E}_{b|1}$ can be estimated by applying a single instrument containing $\mathcal{L}$ as an element~\cite{FootnoteFig}.
} 
}
\label{Fig:overlap}
\end{figure}

\subsection{Characterising complementary instruments}
\CYtwo{
To see how complementarity and incompatibility of instruments are closely related, we first need to characterise the former.
Let $S$ \CYthree{with computational basis \CYthree{$\{\ket{n}\}_{n=0}^{d-1}$ (\mbox{$d<\infty$}})} be $\mathbfcal{E}$'s input system, $S'$ be a system identical to $S$}, and
\mbox{$
\ket{\Phi^+}_{SS'}\coloneqq(1/\sqrt{d})\sum_{n=0}^{d-1}\ket{nn}_{SS'}
$}
be a maximally entangled state
\CYtwo{(\CYthree{since now,} if needed, subscripts denote systems where the objects live in or act on)}.
Then, in \CY{Supplemental Material I,} we prove that 
\CYthree{[below, we define 
\mbox{$\Choi_{a|x}^{(\mathbfcal{E})}\coloneqq(\mathcal{E}_{a|x}\otimes\mathcal{I}_{S'})(\proj{\Phi^+}_{SS'})$}
$\forall\,a,x$]:}

\begin{result}\label{Result:Characterising complementarity}
A pair of instruments $\mathbfcal{E}$ is complementary if and only if 
$
\max_{a,b\in\CYthree{\{0,1,...,d-1\}}}C_{(a,b)}(\mathbfcal{E})=0,
$
where
\begin{align}\label{Eq:SDP}
C_{(a,b)}(\mathbfcal{E})\coloneqq\max\left\{{\rm tr}(P)\,\Big|\,P\ge0,\,\Choi^{(\mathbfcal{E})}_{a|0}\ge P,\,\Choi^{(\mathbfcal{E})}_{b|1}\ge P\right\}.
\end{align}
\end{result}
This thus provides a complete characterisation of complementary instruments.
Importantly, Eq.~\eqref{Eq:SDP} can be solved by SDP and thus can be \CYtwo{checked efficiently.}
Also, if $\max_{a,b}C_{(a,b)}(\mathbfcal{E})$ is positive but sufficiently small, then $\mathbfcal{E}$ is {\em almost} complementary, meaning that it may be very costly to detect the overlap it carries.
%
\CYtwo{Notably, as detailed in \CY{Appendix B}, Result~\ref{Result:Characterising complementarity}} implies that complementary instruments include \CYtwo{complementary POVMs}~\cite{Guhne2023RMP,Hsieh-IP} as special cases. 

\subsection{Complementarity implies incompatibility}
\CYthree{Via Result~\ref{Result:Characterising complementarity}, we can now relate two different faces of uncertainty principle:

}
\begin{result}\label{Result:relation}
If a pair of instruments $\mathbfcal{E}$ is complementary, then it is incompatible; namely,
$
\mathbfcal{E}\notin{\bf JM}.
$
\end{result}
See \CY{Supplemental Material~II} for the proof.
Hence, 
\CYtwo{for instruments,
complementarity is a strong form of incompatibility, consistent with the case of POVMs~\cite{Guhne2023RMP,Starke2026review}.
Result~\ref{Result:relation} implies that complementary instruments can indeed reveal uncertainty principle (see \CY{Appendix C}) and} a direct no-go result: {\em complementary instruments cannot be simulated by \CYtwo{joint measurability}
} [Eq.~\eqref{Eq:JM}].
In fact, as we shall demonstrate later, 
\CYthree{joint measurability can hardly simulate complementarity}
even with the help of quantum pre- or post-processing.

\subsection{Complementarity is strictly stronger than incompatibility for instruments}
Notably, the converse of Result~\ref{Result:relation} is not true.
To see an example, consider the pair of two-outcome qubit instruments
$\mathbfcal{E}^{\rm Pauli}\coloneqq\{\mathcal{E}_{a|x}^{\rm Pauli}\}_{a,x=0,1}$ 
defined by:
\begin{align}\label{Eq:Pauli Z}
&\mathcal{E}^{\rm Pauli}_{0|0}(\cdot)=\proj{0}(\cdot)\proj{0},\;\mathcal{E}^{\rm Pauli}_{1|0}(\cdot)=\proj{1}(\cdot)\proj{1},\\\label{Eq:Pauli X}
&\mathcal{E}^{\rm Pauli}_{0|1}(\cdot)=\proj{+}(\cdot)\proj{+},\;\mathcal{E}^{\rm Pauli}_{1|1}(\cdot)=\proj{-}(\cdot)\proj{-},
\end{align}
where $\ket{\pm}\coloneqq(\ket{0}\pm\ket{1})/\sqrt{2}$.
They are Pauli \CYthree{$Z$ and $X$} measurements, usually viewed as the maximally incompatible pair. 
Now, consider their noisy version,
$\mathbfcal{P}^{(\epsilon)}\coloneqq\big\{\mathcal{P}_{a|x}^{(\epsilon)}\big\}_{a,x}$, with
\begin{align}\label{Eq:noisy Pauli}
\mathcal{P}_{a|x}^{(\epsilon)}(\cdot)\coloneqq(1-\epsilon)\mathcal{E}^{\rm Pauli}_{a|x}(\cdot)+\epsilon\mathcal{D}^{\rm noise}_{a}(\cdot)\quad\forall\,a,x,
\end{align}
where $0\le\epsilon\le1$ indicates the amount of depolarising noise represented by 
the instrument \mbox{$\{\mathcal{D}_{a}^{\rm noise}(\cdot)\coloneqq(\id/4){\rm tr}(\cdot)\}_{a=0,1}$},
namely, 
the quantum output is always maximally mixed $\id/2$, with uniform probability of obtaining each classical outcome.

Result~\ref{Result:Characterising complementarity} implies that $\mathbfcal{P}^{(\epsilon)}$ is complementary if and only if $\epsilon=0$,
meaning that perfect complementary pair can be extremely fragile to noise.
Also, as shown in \CY{Appendix D},
\begin{align}\label{Eq:noisy Pauli JM bound}
\text{\CY{$0\le\epsilon<(\sqrt{2}-1)/\sqrt{2}$}\;
implies \;$\mathbfcal{P}^{(\epsilon)}\notin{\bf JM}$.}
\end{align}
Hence, the entire open interval \CY{$0<\epsilon<(\sqrt{2}-1)/\sqrt{2}$} provides examples of incompatible instruments that are {\em not} complementary, i.e., complementarity is indeed {\em strictly stronger} than incompatibility.
Eq.~\eqref{Eq:noisy Pauli JM bound} suggests that the relation between complementarity and incompatibility is similar to the one between ``pure state entanglement'' and ``mixed state entanglement.''
This also illustrates why complementarity serves as a strong notion of uncertainty principle's signature.

\subsection{No-go theorem on simulating complementarity by \CY{classical-to-quantum signalling}}
Since complementary instruments are all incompatible (Result~\ref{Result:relation}), they cannot be simulated by joint measurability [Eq.~\eqref{Eq:JM}].
A natural question is: 
\begin{center}
{\em What if we give the single instrument more assistants?}
\end{center}
Clearly, by allowing additional operations in Eq.~\eqref{Eq:JM}, one can simulate (possibly many) incompatible instruments.
It is then crucial to know whether complementary instruments, which represent a strong signature of uncertainty principle, are still impossible to simulate.
Inspired by Ref.~\cite{Khandelwal2025PRL}, we consider an upgraded simulation model that incorporates Eq.~\eqref{Eq:JM} and {\em quantum} post-processing.
We say $\mathbfcal{E}$ is {\em jointly measurable via quantum post-processing}, denoted as $\mathbfcal{E}\in{\bf JM}_{\rm post}$, if $\mathbfcal{E}$ can be simulated by a single instrument via quantum and classical post-processing,
\CYthree{i.e.,} if there are channels $\{\mathcal{M}_x\}_x$, conditional probabilities $\{P(a|x,\lambda)\}_{a,x,\lambda}$, and a single instrument $\{\mathcal{G}_\lambda\}_\lambda$ such that
\begin{align}\label{Eq:post-processing}
\mathcal{E}_{a|x} = \sum_\lambda P(a|x,\lambda)\mathcal{M}_x\circ\mathcal{G}_\lambda\quad\forall\,a,x.
\end{align}
This model has a clear physical meaning: It is jointly measurable instruments [Eq.~\eqref{Eq:JM}] assisted by {\em classical-to-quantum signalling}. Namely, we allow the classical input $x$ to alter (i.e., to {\em signal}) the quantum output of the measurement device via some $x$-dependent quantum channel $\mathcal{M}_x$.

Notably, Eq.~\eqref{Eq:post-processing} can simulate many incompatible instruments.
In fact, for every $\mathbfcal{L}\in{\bf JM}$, we have $\sum_a\mathcal{L}_{a|x}=\sum_\lambda\mathcal{G}_\lambda$ $\forall\,x$ for a single instrument $\{\mathcal{G}_\lambda\}_\lambda$. 
Then, for channels $\mathcal{M}_x$'s with 
$(\mathcal{M}_0-\mathcal{M}_1)\circ\sum_\lambda\mathcal{G}_\lambda\neq0$,
in \CY{Appendix E} we show that
\begin{align}\label{Eq:Q post-processing simulation example}
\{\mathcal{M}_x\circ\mathcal{L}_{a|x}\}_{a,x}\in{\bf JM}_{\rm post}\setminus{\bf JM}.
\end{align}
Namely, 
$\{\mathcal{M}_x\circ\mathcal{L}_{a|x}\}_{a,x}$ is incompatible (as it is outside ${\bf JM}$), but it can be simulated by Eq.~\eqref{Eq:post-processing}.
In other words, these incompatible instruments can be ``faked'' by incorporating post-processing quantum channels with a single instrument.


{\em What instruments, if exist, cannot be simulated by Eq.~\eqref{Eq:post-processing}?}
It turns out that complementarity is the key to {\em negate} the simulation model in Eq.~\eqref{Eq:post-processing}.
To see this, first, we say $\mathbfcal{E}$ is {\em rank-one projective} if it is of the form
\begin{align}\label{Eq: Def rank-one projective}
\mathcal{E}_{a|x}(\cdot) = \proj{\psi_a^{(x)}}(\cdot)\proj{\psi_a^{(x)}}\quad\forall\;a,x,
\end{align}
where, for each $x$, the set $\{\ket{\psi_a^{(x)}}\}_a$ is an orthonormal basis.
Such instruments describe the most natural (and noiseless) experimentally feasible measurements.
\CYthree{In fact,} it is already impossible for Eq.~\eqref{Eq:post-processing} to simulate a complementary pair in this subclass.
More precisely, we have that

\begin{result}\label{Result:no-go}
Let $\mathbfcal{E}$ be a pair of rank-one projective instruments.
If $\mathbfcal{E}$ is complementary, then 
$
\mathbfcal{E}\notin{\bf JM}_{\rm post}.
$
\end{result}
See \CY{Supplemental Material III} for the proof.
Result~\ref{Result:no-go} thus tells us that, despite Eq.~\eqref{Eq:post-processing} can indeed simulate lots of incompatible instruments \CYthree{[Eq.~\eqref{Eq:Q post-processing simulation example}]}, strong enough incompatibility (e.g., from rank-one projective complementary instruments) can still demonstrate strong signatures of uncertainty principle that cannot be reproduced by Eq.~\eqref{Eq:post-processing}.
%
Importantly, Result~\ref{Result:no-go} is still valid when the system is subject to noise.
To illustrate this, it suffices to consider an example. 
Let us again use the noisy Pauli instruments $\mathbfcal{P}^{(\epsilon)}$ [Eq.~\eqref{Eq:noisy Pauli}].
Then, as detailed in \CY{Supplemental Material~IV}, we have that

\begin{result}\label{Result:noise robustness}
$0\le\epsilon<2/\left(9+4\sqrt{2}\right)$ implies $\mathbfcal{P}^{(\epsilon)}\notin{\bf JM}_{\rm post}$.
\end{result}
This illustrates that Result~\ref{Result:no-go} is indeed robust to noise.
We remark that $2/(9+4\sqrt{2})\approx0.1365$, and from Eq.~\eqref{Eq:noisy Pauli JM bound} we have 
\CY{$(\sqrt{2}-1)/\sqrt{2}\approx0.2929$}.
\CY{Hence, 
\CYthree{more} than $46\%$ of the interval for the parameter $\epsilon$ to admit incompatible $\mathbfcal{P}^{(\epsilon)}$ \CYthree{[in Eq.~\eqref{Eq:noisy Pauli JM bound}]}
are impossible to simulate via ${\bf JM}_{\rm post}$.}


\subsection{Separating the effects of signalling and uncertainty principle in incompatible instruments}
In fact, Result~\ref{Result:no-go} addresses a 
\CYthree{subtle}
conceptual question about the relation between incompatibility, uncertainty principle, and signalling.
To better illustrate this, consider $\mathbfcal{E}^{\rm Pauli}$ 
\CYthree{[Eqs.~\eqref{Eq:Pauli Z} and~\eqref{Eq:Pauli X}]}.
As we discussed previously, \mbox{$\mathbfcal{E}^{\rm Pauli}\notin{\bf JM}$} [see also Eq.~\eqref{Eq:noisy Pauli JM bound}].
Nevertheless, one can check that \mbox{$\sum_a\mathcal{E}^{\rm Pauli}_{a|0}\neq\sum_a\mathcal{E}^{\rm Pauli}_{a|1}$}; namely, the classical input $x=0,1$ can {\em signal} the average quantum channel.
This classical-to-quantum signalling can already ensure $\mathbfcal{E}^{\rm Pauli}$ is outside ${\bf JM}$ [see Eq.~\eqref{Eq: simple but important fact} in \CY{Appendix~E}]. \CYthree{A conceptual question is:} 
{\em What underpins $\mathbfcal{E}^{\rm Pauli}$'s incompatibility? Is it uncertainty principle, or just signalling?}
Thanks to Result~\ref{Result:no-go}, which tells us \mbox{$\mathbfcal{E}^{\rm Pauli}\notin{\bf JM}_{\rm post}$}, we conclude that joint measurability plus 
\CYthree{classical-to-quantum signalling (modelled by $x$-dependent quantum post-processing)}
are still {\em impossible} to simulate $\mathbfcal{E}^{\rm Pauli}$. 
This
\CYthree{evidences} that incompatibility of $\mathbfcal{E}^{\rm Pauli}$ is 
\CYthree{a consequence of} uncertainty principle.
Hence, introducing the simulation model ${\bf JM}_{\rm post}$ enables us to better identify uncertainty principle's role in incompatible instruments.

\subsection{No-go theorem on simulating complementarity by \CY{pre-measurement manipulation}}
After studying ${\bf JM}_{\rm post}$ with quantum post-processing, it is natural to ask:
{\em What if we instead allow quantum pre-processing to joint measurability?}
Unlike the quantum post-processing, which has a natural interpretation as the extra ability to signal instruments' quantum output, quantum {\em pre-}processing means something different: It is the extra ability to manipulate the instruments' {\em quantum input}.
In this work, we consider unital quantum pre-processing, which covers a wide range of practically relevant channels (e.g., unitary channels and their mixtures).
We say $\mathbfcal{E}$
is {\em jointly measurable via unital quantum pre-processing}, 
denoted as $\mathbfcal{E}\in{\bf JM}_{\rm pre}^{\rm unital}$, if it can be simulated by joint measurability [Eq.~\eqref{Eq:JM}] assisted by
pre-processed unital channels; i.e., there are unital channels $\{\mathcal{N}_x\}_x$, conditional probabilities $\{P(a|x,\lambda)\}_{a,x,\lambda}$, and an instrument $\{\mathcal{G}_\lambda\}_\lambda$ such that
\begin{align}\label{Eq:pre-processing}
\mathcal{E}_{a|x} = \sum_\lambda P(a|x,\lambda)\mathcal{G}_\lambda\circ\mathcal{N}_x\quad\forall\,a,x.
\end{align}
Then, similar to Results~\ref{Result:no-go} and~\ref{Result:noise robustness}, we have the following result:
%

\begin{result}\label{Result:no-go2}
Let $\mathbfcal{E}$ be a pair of rank-one projective instruments.
If $\mathbfcal{E}$ is complementary, then 
$
\mathbfcal{E}\notin{\bf JM}_{\rm pre}^{\rm unital}.
$
Moreover, \mbox{$0\le\epsilon<2/\left(9+4\sqrt{2}\right)$} implies $\mathbfcal{P}^{(\epsilon)}\notin{\bf JM}_{\rm pre}^{\rm unital}$. 
\end{result}
See \CY{Supplemental Material V} for the proof.
Hence, strong enough incompatibility cannot be simulated by joint measurability assisted by quantum pre-processing, and this no-go result is again robust to noise.

Finally, one may ask what if we allow {\em both} $x$-dependent quantum pre- {\em and} post-processing in the simulation model, i.e., 
$
\mathcal{E}_{a|x} = \sum_\lambda P(a|x,\lambda)\mathcal{M}_x\circ\mathcal{G}_\lambda\circ\mathcal{N}_x
$.
We remark that this is already (physically and practically) equivalent to the ability to perform {\em multiple} measurement processes.
For instance, this is one way to implement 
$\mathbfcal{E}^{\rm Pauli}$ experimentally.
Here, we aim to study whether an agent can simulate uncertainty principle's signature {\em without} really having such an ability.
Hence, conceptually, we should not allow this agent to perform {\em both} quantum pre- and post-processing, as it would give the agent this ability, 
thereby trivialising the simulation problem.



\subsection{Operationally separating complementarity and incompatibility for instruments}
Finally, to better understand the difference between complementarity and incompatibility, we ask:
\begin{center}
{\em What operational advantages can only be achieved by complementary instruments?}
\end{center}
Such operational advantages can further help us distinguish incompatibility and complementarity.
To this end, we consider the following task, which aims to unambiguously send one bit of classical information.

Consider a sender ($A$) and a receiver ($B$).
Suppose they pre-share a maximally entangled state $\ket{\Phi_+}_{SS'}$ at the beginning ($A$ and $B$ carry the system $S$ and $S'$, respectively). 
$A$ aims to use the pair of instruments $\mathbfcal{E}$ (acting on $S$) to faithfully send one bit of classical information (encoded in $x=0,1$) to $B$.
In each round, $A$ selects $x$ with probability $\CY{p_x>0}$ and applies the $x$-th instrument from $\mathbfcal{E}$ on the $S$ part of $\ket{\Phi^+}_{SS'}$.
This generates a classical outcome $a$ and a post-instrument quantum state.
$A$ then sends {\em both} to $B$, and $B$ obtains the bipartite state 
$\Choi_{a|x}^{\CY{(\mathbfcal{E})}}/p(a|x)$ with probability $p(a,x)=p_xp(a|x)=\CY{p_x}\tr\big({J}_{a|x}^{\CY{(\mathbfcal{E})}}\big)$ [recall that $\Choi_{a|x}^{\CY{(\mathbfcal{E})}}\coloneqq(\mathcal{E}_{a|x}\otimes\mathcal{I}_{S'})(\proj{\Phi^+}_{SS'})$].

$B$ aims to find a bipartite POVM to correctly recover the index $x$, and $B$ will prefer to simply report ``I don't know'' if the outcome is ambiguous.
For this, conditioned on $a$, $B$ applies a three-outcome POVM $\{Q_0^{(a)},Q_1^{(a)},Q_\emptyset^{(a)}\}$.
When $B$ obtains the outcome $y=0,1$, then $B$ guesses the sent index is $y$.
If $B$ obtains the outcome $\emptyset$, then $B$ waives the round and report ``I don't know.''
Clearly, $B$ can make no mistake 
\CYthree{by always reporting} ``I don't know.''
To avoid trivial cases, we thus consider a fixed parameter $0\le\eta<1$ such that $B$ can only use POVMs with $Q_\emptyset\le\eta\id_{SS'}$. 
Such a POVM has at most $\eta$ probability to produce the outcome $\emptyset$ for {\em all possible} inputs.
$\eta$ thus describes how inconclusive the task is~\cite{Hsieh-IP}.

Now, as detailed in \CY{Appendix F}, after averaging over all $a$'s, the overall 
probability for $B$ to output a wrong index is
\begin{align}\label{Eq:ave error probability_main}
P_{\rm error}^{(\eta)}(\mathbfcal{E})
=
\sum_a\left[\CY{p_0}{\rm tr}\left(Q_1^{(a)}\Choi_{a|0}^{\CY{(\mathbfcal{E})}}\right)+\CY{p_1}{\rm tr}\left(Q_0^{(a)}\Choi_{a|1}^{\CY{(\mathbfcal{E})}}\right)\right].
\end{align}
{\em Unambiguous} communication is achieved if and only if $P_{\rm error}^{(\eta)}(\mathbfcal{E})=0$ for some $0\le\eta<1$,
meaning that $B$ can correctly extract the sent index $x$ with non-vanishing probability.

Now, note that $A$ and $B$ have the freedom to label an instrument's outcomes without changing the physical meaning. 
Mathematically, this is captured by a new instrument pair $\CY{{\bm\pi}(\mathbfcal{E})}\coloneqq\{\mathcal{E}_{\pi_x(a)|x}\}_{a,x}$
with two permutations ${\bm\pi}=\{\pi_0,\pi_1\}$ (one for each $x$ value) used to relabel $\mathbfcal{E}$'s outcomes.
Then, as detailed in \CY{Supplemental Material VI}, we show that
\begin{result}\label{Result:operational meaning} 
$\mathbfcal{E}$ is complementary if and only if there is a parameter $0\le\eta<1$ such that $P_{\rm error}^{(\eta)}\big(\CY{{\bm\pi}(\mathbfcal{E})}\big)=0$ for every ${\bm\pi}$.
\end{result}
We thus obtain a clear operational meaning: $\mathbfcal{E}$ is complementary if and only if no matter how we rename
its classical outcomes, it can always achieve unambiguous classical communication.
This task can tell apart incompatibility and complementarity, and,
once again, telling us
complementarity is strictly stronger than incompatibility.

\section{Discussions}
We identify a novel type of uncertainty principle's signature, termed complementary instruments, which is genuinely stronger than incompatible instruments (Result~\ref{Result:relation}). 
Specifically, we report a series of general noise-robust no-go results: Complementary instruments cannot be simulated by joint measurability assisted by quantum post-processing (classical-to-quantum signalling; Results~\ref{Result:no-go},~\ref{Result:noise robustness}) or unital quantum pre-processing (pre-measurement unital channels; Results~\ref{Result:no-go2}).
Notably, other definitions of complementary instruments have been considered~\cite{Rolino2025arXiv}, which are different from ours.
It is then interesting to relate our results to other types of complementarity, which we leave for future studies.

We also show that whether two instruments are complementary can be completely determined by SDP (Result~\ref{Result:Characterising complementarity}), making it numerically feasible.
Interestingly, solving its dual SDP implies that complementarity is equivalent to the ability to unambiguously send classical information (Result~\ref{Result:operational meaning}).
Combining with the connection between classical communication and energy transfer~\cite{Hsieh2025PRL,Hsieh2025PRA} as well as thermalisation~\cite{Hsieh2021PRXQ,Hsieh2020}, \CYthree{are complementary instruments the necessary resources for} any thermodynamic advantages? 
Further exploration is left for follow-up works.

\section{Acknowledgements}
We thank Yeong-Cherng Liang and Armin Tavakoli for fruitful discussions and comments on the physical interpretation of quantum instruments and their incompatibility.
C.-Y.~H.~acknowledges support from the Leverhulme Trust Early Career Fellowship (``Quantum complementarity: a novel resource for quantum science and technologies'' with Grant No.~ECF-2024-310). 
M.~S. thanks
the National Research Foundation, Singapore through
the National Quantum Office, hosted in A*STAR, under its Centre for Quantum Technologies Funding Initiative (S24Q2d0009); and the Ministry of Education, Singapore, under the Tier 2 grant “Bayesian approach to
irreversibility” (Grant No.~MOE-T2EP50123-0002).
S.-L.~C.~acknowledges the support of the National Science and Technology Council, Taiwan, with Grant No.~NSTC 114-2628-M-005-001-; and National Center for Theoretical Sciences, Taiwan, with Grant No.~NSTC 115-2124-M-002-014-.


\section{Appendix}

\subsection{Appendix A: Incompatibility reveals uncertainty principle}
To see better why incompatibility of instruments is related to the usual notion of uncertainty principle, consider two non-degenerate physical observables $\hat{\mathcal{O}}^{(0)}$ and $\hat{\mathcal{O}}^{(1)}$ with eigenbases $\{\ket{\psi_a^{(x)}}\}_a$ ($x=0,1$), where $a$ labels measurement outcomes.
Then, the process of measuring $\hat{\mathcal{O}}^{(x)}$ is described by the \CYthree{rank-one projective} instrument 
\begin{align}\label{Eq: physical observable induced instruments}
\mathcal{Q}_{a|x}(\cdot)\coloneqq\proj{\CYthree{\phi}_a^{(x)}}(\cdot)\proj{\CYthree{\phi}_a^{(x)}}\quad\forall\,a,x.
\end{align}
Namely, for an initial state $\rho$, we obtain \CYthree{the classical} outcome $a$ \CYthree{and the} post-measurement state $\ket{\CYthree{\phi}_a^{(x)}}$ with the probability $\tr\left[\mathcal{Q}_{a|x}(\rho)\right] = \bra{\psi_a^{(x)}}\rho\ket{\CYthree{\phi}_a^{(x)}}$.
Now, we have the \CYthree{following:}
\begin{fact}\label{fact}
$\{\mathcal{Q}_{a|x}\}_{a,x}\notin{\bf JM}$ implies $\left[\hat{\mathcal{O}}^{(0)},\hat{\mathcal{O}}^{(1)}\right]\neq0$, revealing uncertainty principle.
\end{fact}
Hence, incompatibility can be \CYthree{viewed as} a signature of uncertainty principle (see also Ref.~\cite{Guhne2023RMP}). 
Even more, when measurement processes are noisy (which cannot be captured by physical observables), incompatible instruments can still reveal uncertainty principle's effect via {\em non-joint-measurability.}

{\em Proof of Fact~\ref{fact}.}
If $\left[\hat{\mathcal{O}}^{(0)},\hat{\mathcal{O}}^{(1)}\right]=0$,
then $\hat{\mathcal{O}}^{(0)}$ and $\hat{\mathcal{O}}^{(1)}$ can be simultaneously measured via simultaneous diagonalisation.
\CYthree{Since they are non-degenerate,}
we have \mbox{$\{\ket{\CYthree{\phi}_a^{(0)}}\}_a=\{\ket{\CYthree{\phi}_{\CYthree{a'}}^{(1)}}\}_{\CYthree{a'}}$}, \CYthree{which are the same basis but possibly different orders of labels}.
\CYthree{Namely,}
there is a permutation $\pi$ achieving $\ket{\CYthree{\phi}_{\pi(a)}^{(0)}} = \ket{\CYthree{\phi}_a^{(1)}}$ $\forall\,a$.
This further implies that 
\begin{align}
\mathcal{Q}_{\pi(a)|0}=\mathcal{Q}_{a|1}\quad\forall\,a,
\end{align}
which form a compatible pair of instruments.
To see this, let us choose a binary variable $\lambda=0,1$ with the instrument \mbox{$\mathcal{G}_\lambda=\mathcal{Q}_{\lambda|0}$} and \CYthree{the} conditional probability 
\begin{align}
\CYthree{P(a|x,\lambda)=\delta_{a,\lambda}\delta_{x,0}+\delta_{\pi(a),\lambda}\delta_{x,1}.}
\end{align}
Then, one can check that 
\begin{align}
\mathcal{Q}_{a|x}=\sum_{\lambda}P(a|,x,\lambda)\mathcal{G}_\lambda\quad\forall\,a,x.
\end{align}
This thus means $\{\mathcal{Q}_{a|x}\}_{a,x}\in{\bf JM}$.
Consequently, if $\{\mathcal{Q}_{a|x}\}_{a,x}\notin{\bf JM}$, then we must have $\left[\hat{\mathcal{O}}^{(0)},\hat{\mathcal{O}}^{(1)}\right]\neq0$.
\hfill$\square$

\subsection{Appendix B: Complementary POVMs as special cases of complementary instruments}
\CYtwo{When $\mathbfcal{E}$'s output space} has dimension one (denoted by $d_{\rm out}=1$), \CYtwo{Result~\ref{Result:Characterising complementarity}} implies that complementary instruments are equivalent to complementary POVMs.
\CYtwo{To see this,} recall that a pair of POVMs $\{E_{a|x}\}_{a,x}$ (namely, for each $x=0,1$, the set $\{E_{a|x}\}_a$ is a POVM) is said to be {\em complementary}~\cite{Guhne2023RMP,Hsieh-IP} if the following statement is true for {\em every} pair of outcome index pair $(a,b)$ and positive semi-definite operator $P\ge0$:
\begin{align}\label{Eq:Complementary POVMs}
E_{a|0}\ge P\;\;\&\;\;E_{b|1}\ge P\;\;\Longrightarrow\;\;P=0.
\end{align}
\CYtwo{
The operational meaning is similar to the one discussed in Fig.~\ref{Fig:overlap}, but now it is only for classical outcomes.
Namely, the left-hand side of Eq.~\eqref{Eq:Complementary POVMs} implies that 
\begin{align}
{\rm tr}(E_{a|0}\rho)\ge{\rm tr}(P\rho)\quad\&\quad{\rm tr}(E_{b|1}\rho)\ge{\rm tr}(P\rho)\quad\forall\,\rho,
\end{align}
where $P$ can be viewed as an element of a two-outcome POVM $\{\id-P,P\}$.
Hence, any non-vanishing $P\neq0$ achieving this implies that non-trivial ``common parts'' can be extracted by implementing a single POVM $\{\id-P,P\}$.
Eq.~\eqref{Eq:Complementary POVMs} thus means that two POVMs are complementary if no such common parts can exist.
}

Now, for
$\mathbfcal{E}=\{\mathcal{E}_{a|x}\}_{a,x}$ with $d_{\rm out}=1$, we
can express it as
\begin{align}\label{Eq:Instruments with trivial output space}
\mathcal{E}_{a|x}={\rm tr}[E_{a|x}(\cdot)]\quad\forall\,a,x,
\end{align}
where $\{E_{a|x}\}_{a,x}$ is a set of POVMs acting on \CYthree{$\mathbfcal{E}$'s input system $S$} (this is actually the inverse Choi map~\cite{Choi1975,Jamiolkowski1972} of $\mathcal{E}_{a|x}$).
Then, we have:
\begin{fact}\label{fact: complementary POVMs}
When $d_{\rm out}=1$, $\mathbfcal{E}$ is complementary if and only if $\{E_{a|x}\}_{a,x}$ is complementary, and they are related via Eq.~\eqref{Eq:Instruments with trivial output space}.
\end{fact}
\begin{proof}
Using $(\id_S\otimes P^t_{S'})\ket{\Phi^+}_{SS'} = (P_S\otimes \id_{S'})\ket{\Phi^+}_{SS'}$ for every normal operator $P$, where $(\cdot)^t$ is the transpose operation with respect to the given computational basis, we have
\begin{align}
(\mathcal{E}_{a|x}\otimes\mathcal{I}_{\CYtwo{S'}})(\proj{\Phi^+}_{\CYtwo{SS'}})=E_{a|x}^t/d_{\rm in}\quad\forall\,a,x,
\end{align}
where $d_{\rm in}$ is the dimension of \CYtwo{$S$ and $S'$.}
Using \CYtwo{Result~\ref{Result:Characterising complementarity}} [see also Lemma~\ref{lemma: Choi cahracterisation common Q part} and Eq.~\eqref{Eq:Complementary instruments defining condition} in \CY{Supplemental Material I}], $\mathbfcal{E}$ is complementary if and only if the following statement is true for every outcome index pair $(a,b)$ and positive semi-definite operator $P\ge0$:
\begin{align}
E_{a|0}^t/d_{\rm in}\ge P\;\;\&\;\;E_{b|1}^t/d_{\rm in}\ge P\;\;\Longrightarrow\;\;P=0.
\end{align}
This is equivalent to demanding Eq.~\eqref{Eq:Complementary POVMs} for every outcome index pair $(a,b)$ and positive semi-definite operator $P\ge0$.
Hence, $\mathbfcal{E}$ is complementary if and only if $\{E_{a|x}\}_{a,x}$ [related via Eq.~\eqref{Eq:Instruments with trivial output space}] is complementary.
\end{proof}
Fact~\ref{fact: complementary POVMs} implies that complementary instruments serve as a more general notion that contains all complementary POVMs as special cases, and the latter are equivalent to complementary instruments with a trivial quantum output space.

\subsection{Appendix C: Complementarity reveals uncertainty principle}
Result~\ref{Result:relation} shows how complementary instruments can reveal uncertainty principle.
For non-degenerate physical observables $\hat{\mathcal{O}}^{(0)}$ and $\hat{\mathcal{O}}^{(1)}$ with eigenbases $\{\ket{\CYthree{\phi}_a^{(x)}}\}_a$ (with $x=0,1$), as we have discussed \CY{in Appendix A}, the incompatibility of the induced instruments $\mathcal{Q}_{a|x}(\cdot)=\proj{\CYthree{\phi}_a^{(x)}}(\cdot)\proj{\CYthree{\phi}_a^{(x)}}$ as given in Eq.~\eqref{Eq: physical observable induced instruments} implies $[\hat{\mathcal{O}}^{(0)},\hat{\mathcal{O}}^{(1)}]\neq0$.
Then Result~\ref{Result:relation} implies that the complementarity of $\{\mathcal{Q}_{a|x}\}_{a,x}$ implies non-commutativity of $\hat{\mathcal{O}}^{(0)}$ and $\hat{\mathcal{O}}^{(1)}$ (see also Ref.~\cite{Starke2026review}), \CYtwo{thereby revealing uncertainty principle.}

\subsection{Appendix D: Proof of Eq.~\eqref{Eq:noisy Pauli JM bound}}
For $\mathbfcal{P}^{(\epsilon)}$, let us consider the set of sub-normalised states 
\begin{align}
\mathbfcal{P}^{(\epsilon)}(\id/2)\coloneqq\{\mathcal{P}_{a|x}^{(\epsilon)}(\id/2)\}_{a,x}.
\end{align}
\CYthree{This set is said to admit a} {\em local hidden-state} (LHS) model~\cite{Wiseman2007PRL,Jones2007PRA,Uola2020RMP,Cavalcanti2016,Xiang2022PRXQuantum}, denoted as $\mathbfcal{P}^{(\epsilon)}(\id/2)\in{\bf LHS}$, if it can be written in the following form:
\begin{align}\label{Eq:LHS}
\mathcal{P}_{a|x}^{(\epsilon)}\left({\id/2}\right)\stackrel{\rm LHS}{=}\sum_\lambda P(a|x,\lambda)q(\lambda)\rho_\lambda\quad\forall\,a,x,
\end{align}
where $\lambda$ is the so-called hidden variable, $\{P(a|x,\lambda)\}_{a,x,\lambda}$ as well as $\{q(\lambda)\}_\lambda$ are (conditional) probability distributions, and \CYthree{$\rho_\lambda$ are states}. 
Then, $\mathbfcal{P}^{(\epsilon)}(\id/2)$ is called {\em steerable}~\cite{Wiseman2007PRL,Jones2007PRA,Uola2020RMP,Cavalcanti2016,Xiang2022PRXQuantum} if it cannot be described by any LHS model, namely, $\mathbfcal{P}^{(\epsilon)}(\id/2)\notin{\bf LHS}$.
Quantum steering is an essential resource for violating Bell inequalities and one-sided device-independent quantum information tasks (see, e.g., Refs.~\cite{Uola2020RMP,Cavalcanti2016,Xiang2022PRXQuantum} for a brief review), which are beyond the scope of this work.
Here, we only use it as a tool
to prove Eq.~\eqref{Eq:noisy Pauli JM bound}.

From the definition of {\bf LHS}, one can directly see that it is closely related to joint measurability defined in Eq.~\eqref{Eq:JM}.
In fact, when $\mathbfcal{P}^{(\epsilon)}\in{\bf JM}$, we can write \CYtwo{(see Ref.~\cite{Hsieh2024PRL} and the discussion above Eq.~(4) therein)}
\begin{align}
\CYtwo{\mathcal{P}_{a|x}^{(\epsilon)}(\id/2)=\sum_\lambda P(a|x,\lambda)\mathcal{G}_\lambda(\id/2)\quad\forall\,a,x,} 
\end{align}
which becomes a LHS model by setting the probability \mbox{$q(\lambda)={\rm tr}[\mathcal{G}_\lambda(\id/2)]$} and the state $\rho_\lambda=\mathcal{G}_\lambda(\id/2)/q(\lambda)$.
Namely, 
$\mathbfcal{P}^{(\epsilon)}\in{\bf JM}$ implies $\mathbfcal{P}^{(\epsilon)}(\id/2)\in{\bf LHS}$.
This thus implies \CYthree{that}
%
\begin{align}\label{fact: steering incomp}
\text{If $\mathbfcal{P}^{(\epsilon)}(\id/2)$ is steerable, then $\mathbfcal{P}^{(\epsilon)}$ is incompatible.}
\end{align}
Hence, the problem \CYthree{of certifying $\mathbfcal{P}^{(\epsilon)}(\id/2)$'s incompatibility} becomes a steering certification \CYthree{problem.}
To certify steering, we follow the approach of Ref.~\cite{Hsieh2024PRL} [see its Supplemental Material IX Eqs.~(120-125)] and consider the following steering inequality shown in Ref.~\cite{Pusey2013} (which is originated from \CYthree{Eq.~(63) in} Ref.~\cite{Cavalcanti2009}), which, in our current context, reads:
\begin{align}\label{Eq:steering inequality}
\sum_{a,x}{\rm tr}\left[F_{a|x}\mathcal{P}_{a|x}^{(\epsilon)}(\id/2)\right]\stackrel{\rm LHS}{\le}\sqrt{2}\quad{\rm if}\;\mathbfcal{P}^{(\epsilon)}(\id/2)\in{\bf LHS},
\end{align}
where 
\begin{align}
&F_{0|0}=-F_{1|0}=\proj{0}-\proj{1};\\
&F_{0|1}=-F_{1|1}=\proj{+}-\proj{-},
\end{align} 
namely, they are Pauli \CYthree{Z and X} with a plus/minus sign.
Consequently, a violation of the inequality in Eq.~\eqref{Eq:steering inequality} implies steering of $\mathbfcal{P}^{(\epsilon)}(\id/2)$, and hence the incompatibility of $\mathbfcal{P}^{(\epsilon)}$ due to Eq.~\eqref{fact: steering incomp}.
Now, as one can check, for the noiseless Pauli instruments $\mathbfcal{E}^{\rm Pauli}$ [Eqs.~\eqref{Eq:Pauli Z} and~\eqref{Eq:Pauli X}], we have
\begin{align}
\sum_{a,x}{\rm tr}\left[F_{a|x}\mathcal{E}_{a|x}^{\rm Pauli}(\id/2)\right]=2.
\end{align}
Also, $\sum_{a,x}{\rm tr}\left[F_{a|x}(\id/4)\right]=0$.
Using Eq.~\eqref{Eq:noisy Pauli}, we obtain
\begin{align}
\sum_{a,x}{\rm tr}\left[F_{a|x}\mathcal{P}_{a|x}^{(\epsilon)}(\id/2)\right]=2(1-\epsilon),
\end{align}
and the inequality in Eq.~\eqref{Eq:steering inequality} is violated if $1-\epsilon>1/\sqrt{2}$, which leads to the desired bound in Eq.~\eqref{Eq:noisy Pauli JM bound}.
\hfill$\square$

\subsection{Appendix E: Proof of Eq.~\eqref{Eq:Q post-processing simulation example}}
For any $\mathbfcal{L}\in{\bf JM}$ and channels $\mathcal{M}_x$'s, by definition the set $\{\mathcal{M}_x\circ\mathcal{L}_{a|x}\}_{a,x}$ is a well-defined pair of instruments, and is of the form Eq.~\eqref{Eq:post-processing}; namely, $\{\mathcal{M}_x\circ\mathcal{L}_{a|x}\}_{a,x}\in{\bf JM}_{\rm post}$.
Now, note that, generally, 
\begin{align}\label{Eq: simple but important fact}
\text{
$\mathbfcal{E}\in{\bf JM}$ must implies $\sum_a\mathcal{E}_{a|0}=\sum_a\mathcal{E}_{a|1}$,
}
\end{align}
which can be directly seen from Eq.~\eqref{Eq:JM}.
Consequently, every $\mathbfcal{E}$ with $\sum_a\mathcal{E}_{a|0}\neq\sum_a\mathcal{E}_{a|1}$ must be incompatible.
By setting $\mathcal{E}_{a|x}=\mathcal{M}_x\circ\mathcal{L}_{a|x}$, the assumed condition \mbox{$(\mathcal{M}_0-\mathcal{M}_1)\circ\sum_\lambda\mathcal{G}_\lambda\neq0$} together with the fact \mbox{$\sum_a\mathcal{L}_{a|x}=\sum_\lambda\mathcal{G}_\lambda$ $\forall\,x$}
implies that
\begin{align}
\mathcal{M}_0\circ\sum_a\mathcal{L}_{a|0}&=\mathcal{M}_0\circ\sum_\lambda\mathcal{G}_\lambda\nonumber\\
&\neq\mathcal{M}_1\circ\sum_\lambda\mathcal{G}_\lambda=\mathcal{M}_1\circ\sum_a\mathcal{L}_{a|1}.
\end{align} 
Hence, the pair of instruments $\{\mathcal{M}_x\circ\mathcal{L}_{a|x}\}_{a,x}$ is incompatible, meaning that it is outside ${\bf JM}$.
\hfill$\square$

\subsection{Appendix F: Calculation for Eq.~\eqref{Eq:ave error probability_main}}
\CYthree{We note that,} {\em conditioned on $a$,} what $B$ needs to do is distinguishing two states $\{{J}_{a|x}^{\CY{(\mathbfcal{E})}}/p(a|x)\}_{x=0,1}$ subject to the probability 
\begin{align}
p(x|a)\eqqcolon\frac{p(a,x)}{\CY{p'_a}}=\frac{p(a|x)\CY{p_x}}{\CY{p'_a}},
\end{align}
where \CYthree{$p'_a\coloneqq p(a,0)+p(a,1)$} is the probability of obtaining $a$.
For this $a$ index, the probability of reporting a wrong answer (i.e., $y\neq x$) reads
\begin{align}
&P_{{\rm error}|a}^{(\eta)}\coloneqq p(0|a){\rm tr}\left[Q_1^{(a)}\frac{\Choi_{a|0}^{\CY{(\mathbfcal{E})}}}{p(a|0)}\right]+p(1|a){\rm tr}\left[Q_0^{(a)}\frac{\Choi_{a|1}^{\CY{(\mathbfcal{E})}}}{p(a|1)}\right]\nonumber\\
&=\frac{\CY{p_0}{\rm tr}\left(Q_1^{(a)}\Choi_{a|0}^{\CY{(\mathbfcal{E})}}\right)+\CY{p_1}{\rm tr}\left(Q_0^{(a)}\Choi_{a|1}^{\CY{(\mathbfcal{E})}}\right)}{\CY{p'_a}}.
\end{align}
Averaging over all $a$'s, the overall error probability reads 
\begin{align}\label{Eq:ave error probability}
P_{\rm error}^{(\eta)}(\mathbfcal{E}) &\coloneqq \sum_a\CY{p'_a}P_{{\rm error}|a}^{(\eta)}\nonumber\\
&=
\sum_a\left[\CY{p_0}{\rm tr}\left(Q_1^{(a)}\Choi_{a|0}^{\CY{(\mathbfcal{E})}}\right)+\CY{p_1}{\rm tr}\left(Q_0^{(a)}\Choi_{a|1}^{\CY{(\mathbfcal{E})}}\right)\right],
\end{align}
which gives Eq.~\eqref{Eq:ave error probability_main}.

\clearpage
\section{Supplemental Material}

\CY{
\subsection{Supplemental Material I: Proof of Result~\ref{Result:Characterising complementarity}}
}

In this section, we provide a detailed proof of 
\CYtwo{Result~\ref{Result:Characterising complementarity} in the main text.}
%
Just like in the main text, for a $d$-dimensional system $S$, we use $S'$ to denote a system identical to $S$. 
In the bipartite system $SS'$, 
\begin{align}\label{Eq: max ent state}
\ket{\Phi^+}_{SS'}\coloneqq\frac{1}{\sqrt{d}}\sum_{n=0}^{d-1}\ket{nn}_{SS'}
\end{align} 
is a maximally entangled state, where $\{\ket{n}\}_n$ is the given computational basis in $S$ and $S'$.
For simplicity, we write \mbox{$\Phi_{SS'}^+\coloneqq\proj{\Phi^+}_{SS'}$}.
\CYtwo{To better specify input and output spaces, since now, a map is said to be an {\em $\Sin\to\Sout$ filter/channel} if it is a filter/channel with input space $\Sin$ and output space $\Sout$, where $\Sin$ and $\Sout$ are quantum systems with finite dimensions $\dSin$ and $\dSout$, respectively.
Likewise, $\mathbfcal{E}$ is called a set of {\em $\Sin\to\Sout$ instruments} if all its filters $\mathcal{E}_{a|x}$ are $\Sin\to\Sout$ filters.  
Finally, $\Sin'$ ($\Sout'$) denotes a system identical to $\Sin$ ($\Sout$).}

\subsubsection{Characterisation Lemmas}
Our goal now is to fully characterise complementary instruments with arbitrary finite input and output dimensions $\dSin$ and $\dSout$.
To this end, we need to develop the following lemma.
Below, $A$ and $B$ denote generic quantum systems with finite dimensions $d_A$ and $d_B$, respectively.

\begin{lemma}\label{lemma: P_{SS'} filter}
Let $P_{AB}\ge0$ with $P_{AB}\neq0$ be an operator acting on the bipartite system $AB$.
Then there exist a number $\alpha>0$ and a \CY{non-vanishing $B\to A$ filter $\mathcal{L}\neq0$} achieving
\begin{align}
P_{AB'} = \alpha(\mathcal{L}\otimes\mathcal{I}_{B})(\Phi^+_{BB'}),
\end{align}
where $B'$ is a system identical to $B$.
\end{lemma}
\begin{proof}
First, recall that, for any pure state $\ket{\psi}_{AB'}$ in $AB'$, we can write down its Schmidt decomposition (see, e.g., Theorem 2.7 and Exercise 2.76 in Ref.~\cite{QIC-book}) as 
\begin{align}
\ket{\psi}_{AB'}=\sum_nq_n(U_A\otimes V_{B'})(\ket{n}_{A}\otimes\ket{n}_{B'}), 
\end{align}
where $\ket{n}_A$ and $\ket{n}_{B'}$ are from the given computational bases in $A$ and $B'$ used to define $\ket{\Phi^+}_{AA'}$ and $\ket{\Phi^+}_{BB'}$ via Eq.~\eqref{Eq: max ent state}. 
Also, $U_A$ and $V_{B'}$ are some single-party unitary operators acting on $A$ and $B'$, respectively. 
Finally, $q_n\ge0$ satisfying $\sum_nq_n^2=1$ are the so-called Schmidt coefficients~\cite{QIC-book}. 
Now, define the operator
\begin{align}
Q\coloneqq\sum_nq_nU_A\ket{n}_A\bra{n}_B,
\end{align}
which maps from $B$ to $A$, and we keep its system dependence implicit for clarity.
Then, one can rewrite $\ket{\psi}_{AB'}$ as
\begin{align}
\ket{\psi}_{AB'}&= (Q\otimes V_{B'})\sum_n\ket{n}_{B}\otimes\ket{n}_{B'}\nonumber\\
&=\left[\left(\sqrt{d_B}QV_B^t\right)\otimes\id_{B'}\right]\ket{\Phi^+}_{BB'}.
\end{align}
Here, we have used 
\mbox{$(\id_B\otimes V_{B'})\ket{\Phi^+}_{BB'} = (V^t_B\otimes \id_{B'})\ket{\Phi^+}_{BB'}$}, where $(\cdot)^t$ is the transpose operation with respect to the 
basis $\{\ket{n}_B\}_n$.
Now, define the following mapping from $B$ to $A$:
\begin{align}
\mathcal{F}(\cdot)\coloneqq QV^t(\cdot)\left(QV^t\right)^\dagger,
\end{align}
which directly implies that
\begin{align}\label{Eq:purestatefilter}
\proj{\psi}_{AB'} = d_B(\mathcal{F}\otimes\mathcal{I}_{B'})(\Phi^+_{BB'}).
\end{align}
Hence, by Choi-Jamio\l kowski isomorphism~\cite{Choi1975,Jamiolkowski1972}, 
$\mathcal{F}$ is a completely-positive linear map.
Moreover, this is a mapping with a single Kraus operator $QV^t$. 
Since
\begin{align}
\left(QV^t\right)^\dagger QV^t &= V^{t,\dagger} Q^\dagger QV^t\nonumber\\
&= \sum_nq_n^2V^{t,\dagger}\ket{n}_B\bra{n}_BV^t\le\id_B,
\end{align}
we have that
\begin{align}
{\rm tr}[\mathcal{F}(\rho)]={\rm tr}\left[\left(QV^t\right)^\dagger QV^t\rho\right]\le{\rm tr}(\rho)\quad\forall\,\rho:\text{state}.
\end{align}
This means $\mathcal{F}$ is also trace-non-increasing~\cite{QIC-book}; namely, it is a valid filter.
Consequently, 
for {\em every} pure state $\ket{\psi}_{AB'}$, there is a $B\to A$ filter $\mathcal{F}$ achieving Eq.~\eqref{Eq:purestatefilter}.

Now we can proceed to complete the proof of the lemma.
Since $P_{AB'}\neq0$ and $P_{AB'}\ge0$, we can write 
\begin{align}\label{Eq: P induced state}
\rho_{AB'}\coloneqq \frac{P_{AB'}}{{\rm tr}(P_{AB'})}, 
\end{align}
which is a normalised state in the bipartite system $AB'$, and hence has a spectral decomposition as~\cite{QIC-book}
\begin{align}
\rho_{AB'} = \sum_lp_l\proj{\psi_l}_{AB'}
\end{align}
for some bipartite pure states $\ket{\psi_l}_{AB'}$ and eigenvalues \mbox{$p_l\ge0$} satisfying $\sum_lp_l=1$.
Using Eq.~\eqref{Eq:purestatefilter}, for each $l$, there is a $B\to A$ filter $\mathcal{F}^{(l)}$ achieving 
\begin{align}
\proj{\psi_l}_{AB'} = d_B(\mathcal{F}^{(l)}\otimes\mathcal{I}_{B'})(\Phi^+_{BB'}).
\end{align}
Consequently, we have that
\begin{align}
\rho_{AB'} &= d_B\left[\left(\sum_lp_l\mathcal{F}^{(l)}\right)\otimes\mathcal{I}_{B'}\right](\Phi^+_{BB'})\nonumber\\
&\eqqcolon d_B(\mathcal{L}\otimes\mathcal{I}_{B'})(\Phi^+_{BB'}),
\end{align}
where
\begin{align}
\mathcal{L}\coloneqq\sum_lp_l\mathcal{F}^{(l)}
\end{align}
is again a $B\to A$ filter.
Finally, define
\begin{align}
\alpha\coloneqq d_B{\rm tr}(P_{BB'})>0.
\end{align}
Then, together with Eq.~\eqref{Eq: P induced state}, we have 
\begin{align}
P_{AB'} = d_B{\rm tr}(P_{AB'})\times\frac{P_{AB'}}{d_B{\rm tr}(P_{AB'})} = \alpha(\mathcal{L}\otimes\mathcal{I}_{B'})(\Phi^+_{BB'}),
\end{align}
which thus completes the proof.
\end{proof}

Using the above lemma, we can prove the following characterisation of instruments that \CY{have non-vanishing overlap}:

\begin{lemma}\label{lemma: Choi cahracterisation common Q part}
A pair of $\Sin\to\Sout$ instruments $\mathbfcal{E}$ \CY{has non-vanishing overlap} if and only if there exist an outcome index pair $(a,b)$ and an operator $P_{\Sout\Sin'}\ge0$ with $P_{\Sout\Sin'}\neq0$ such that
\begin{align}\label{Eq: lemma condition001}
&(\mathcal{E}_{a|0}\otimes\mathcal{I}_{\Sin'})(\Phi^+_{\Sin\Sin'})\ge P_{\Sout\Sin'};\\\label{Eq: lemma condition002}
&(\mathcal{E}_{b|1}\otimes\mathcal{I}_{\Sin'})(\Phi^+_{\Sin\Sin'})\ge P_{\Sout\Sin'}.
\end{align}
\end{lemma}
\begin{proof}
\CYtwo{By definition [Eq.~\eqref{Eq:common part 01} in the main text]},
$\mathbfcal{E}$ has non-vanishing overlap if and only if 
\CYtwo{
$\mathcal{E}_{a|0}-\mathcal{L}$ and $\mathcal{E}_{b|1}-\mathcal{L}$ are both completely-positive
for some outcome index pair $(a,b)$ and non-vanishing $\Sin\to\Sout$ filter $\mathcal{L}\neq0$.}
This is true if and only if, via the Choi-Jamio\l kowski isomorphism~\cite{Choi1975,Jamiolkowski1972}, 
\begin{align}\label{Eq: Choi common Q inf def001}
&[(\mathcal{E}_{a|0}-\CYtwo{\mathcal{L}})\otimes\mathcal{I}_{\Sin'}](\Phi^+_{\Sin\Sin'})\ge 0;\\\label{Eq: Choi common Q inf def002}
&[(\mathcal{E}_{b|1}-\CYtwo{\mathcal{L}})\otimes\mathcal{I}_{\Sin'}](\Phi^+_{\Sin\Sin'})\ge 0
\end{align}
\CYtwo{for, again, some outcome index pair $(a,b)$ and non-vanishing $\Sin\to\Sout$ filter $\mathcal{L}\neq0$.}
We thus obtain a statement that is necessary and sufficient to $\mathbfcal{E}$ \CY{having non-vanishing overlap.}
With the above observation, we can now prove the lemma.

({\em Proof of the direction ``$\Rightarrow$''}) If $\mathbfcal{E}$ \CY{has non-vanishing overlap}, then the result follows by setting 
\begin{align}
P_{\Sout\Sin'} = (\CYtwo{\mathcal{L}}\otimes\mathcal{I}_{\Sin'})(\Phi^+_{\Sin\Sin'})\ge0, 
\end{align}
which is non-vanishing if $\mathbfcal{E}$ \CY{has non-vanishing overlap.}

({\em Proof of the direction ``$\Leftarrow$''}) Suppose Eqs.~\eqref{Eq: lemma condition001} and~\eqref{Eq: lemma condition002} hold for some outcome index pair $(a,b)$ and $P_{\Sout\Sin'}\ge0$ with $P_{\rm \Sout\Sin'}\neq0$.
Then, Lemma~\ref{lemma: P_{SS'} filter} implies that there exist a strictly positive value $\alpha>0$ and a non-vanishing $\Sin\to\Sout$ filter \CYtwo{$\widetilde{\mathcal{L}}\neq0$} achieving 
\begin{align}
P_{\Sout\Sin'} = \alpha(\CYtwo{\widetilde{\mathcal{L}}}\otimes\mathcal{I}_{\Sin'})(\Phi^+_{\Sin\Sin'}).
\end{align}
Now, we define
\begin{align}
\mathcal{L}\coloneqq\begin{cases}\;\alpha\widetilde{\mathcal{L}}\quad&{\rm if}\;\alpha\le1\\\;\widetilde{\mathcal{L}}\quad&{\rm if}\;\alpha>1,\end{cases}
\end{align}
which is a non-vanishing $\Sin\to\Sout$ filter achieving
\begin{align}
P_{\Sout\Sin'} \ge (\mathcal{L}\otimes\mathcal{I}_{\Sin'})(\Phi^+_{\Sin\Sin'}).
\end{align}
Substituting this into Eqs.~\eqref{Eq: lemma condition001} and~\eqref{Eq: lemma condition002} will give us Eqs.~\eqref{Eq: Choi common Q inf def001} and~\eqref{Eq: Choi common Q inf def002}, meaning that $\mathbfcal{E}$ indeed \CY{has non-vanishing overlap}.
The proof is thus completed.
\end{proof}

The above lemma thus implies \CYtwo{Result~\ref{Result:Characterising complementarity}} as a corollary:

\subsubsection{Proof of \CYtwo{Result~\ref{Result:Characterising complementarity}}}
\begin{proof}
Let us use the notation defined in the main text 
\begin{align}\label{Eq:Choi operator diff input output}
\Choi_{a|x}^{(\mathbfcal{E})}\coloneqq(\mathcal{E}_{a|x}\otimes\mathcal{I}_{\Sin'})(\Phi^+_{\Sin\Sin'})\quad\forall\,a,x.
\end{align}
From Lemma~\ref{lemma: Choi cahracterisation common Q part}, a pair of $\Sin\to\Sout$ instruments $\mathbfcal{E}$ is complementary (i.e., \CY{has {\em vanishing} overlap}) if and only if the following statement is true for {\em every} outcome index pair $(a,b)$ and {\em every} positive semi-definite operator $P_{\Sout\Sin'}\ge0$:
\begin{align}\label{Eq:Complementary instruments defining condition}
\Choi_{a|0}^{(\mathbfcal{E})}\ge P_{\Sout\Sin'}\;\;\&\;\;
\Choi_{b|1}^{(\mathbfcal{E})}\ge P_{\Sout\Sin'}
\;\Longrightarrow\;\CY{P_{\Sout\Sin'}=0}.
\end{align}
\CYtwo{This thus proves Result~\ref{Result:Characterising complementarity}.}
\end{proof}

\subsection{Supplemental Material II: Proof of Result~\ref{Result:relation}}
\begin{proof}
It suffices to show that compatibility [Eq.~\eqref{Eq:JM} in the main text] implies \CY{non-vanishing overlap}
\CYtwo{[Eq.~\eqref{Eq:common part 01} in the main text]}.
Suppose $\mathbfcal{E}$ \CYtwo{(with finite-dimensional input and output systems $\Sin$ and $\Sout$, respectively)} is compatible, meaning that we can write \mbox{$\mathcal{E}_{a|x} = \sum_\lambda P(a|x,\lambda)\mathcal{G}_\lambda$} $\forall\,a,x$ for an \CYtwo{$\Sin\to\Sout$} instrument $\{\mathcal{G}_\lambda\}_\lambda$ and conditional probability distributions $\{P(a|x,\lambda)\}_{a,x,\lambda}$.
Then, applying it on half of the maximally entangled state $\ket{\Phi^+}_{\CYtwo{\Sin\Sin'}}$ [\CYtwo{with $S=\Sin$ in} Eq.~\eqref{Eq: max ent state}], we obtain
\begin{align}\label{Eq:result 2 proof 001}
(\mathcal{E}_{a|x}\otimes\mathcal{I}_{\CYtwo{\Sin'}})(\Phi_{\CYtwo{\Sin\Sin'}}^+) = \tau_{a|x}^{\rm LHS}\quad\forall\,a,x,
\end{align}
where
\begin{align}\label{Eq:clasiscal model}
\tau_{a|x}^{\rm LHS}\coloneqq\sum_\lambda P(a|x,\lambda)(\mathcal{G}_\lambda\otimes\mathcal{I}_{\CYtwo{\Sin'}})(\Phi^+_{\CYtwo{\Sin\Sin'}})\quad\forall\,a,x
\end{align}
are described by the {\em local-hidden-state} (LHS) model~\cite{Wiseman2007PRL,Jones2007PRA,Uola2020RMP,Cavalcanti2016,Xiang2022PRXQuantum} \CYtwo{in the bipartite system $\Sout\Sin'$} [see also Eq.~\eqref{Eq:LHS} in \CY{Appendix D}]. Via a new local-hidden variable $\mu$, this can be rewritten into~\cite{Uola2020RMP,Cavalcanti2016} (see also Fact 2 in Supplemental Material of Ref.~\cite{Hsieh2024PRL})
\begin{align}\label{Eq:post-mu}
\tau_{a|x}^{\rm LHS} = \sum_\mu D(a|x,\mu)\omega_{\CYtwo{\Sout\Sin'}}^{(\mu)}\quad\forall\,a,x,
\end{align}
where $D(a|x,\mu)$'s are deterministic probability distributions (which only take the value either $0$ or $1$), and $\omega_{\CYtwo{\Sout\Sin'}}^{(\mu)}\ge0$ are sub-normalised state such that $\sum_\mu\omega_{\CYtwo{\Sout\Sin'}}^{(\mu)}$ is a state.

Crucially, for every given $\mu$, there must exist two outcome indices $a_0{(\mu)},a_1{(\mu)}$ achieving 
\begin{align}\label{Eq: each mu has at least one a index}
D(a_0{(\mu)}|0,\mu) = 1 = D(a_1{(\mu)}|1,\mu);
\end{align}
that is, we have $D(a_x(\mu)|x,\mu)=1$ for $x=0,1$.
Also, there must exist at least one value of $\mu$ (denote it by $\mu_*$) achieving $\omega_{\CYtwo{\Sout\Sin'}}^{(\mu_*)}\neq0$ (if not, we then have $\tau_{a|x}^{\rm LHS}=0$ $\forall\,a,x$, implying $\mathcal{E}_{a|x}=0$ $\forall\,a,x$, a contradiction).
Combining everything together, we thus obtain, for $x=0,1$,
\begin{align}
&(\mathcal{E}_{a_x{(\mu_*)}|x}\otimes\mathcal{I}_{\CYtwo{\Sin'}})(\Phi_{\CYtwo{\Sin\Sin'}}^+) = \sum_\mu D(a_x{(\CYtwo{\mu_*})}|x,\mu)\omega_{\CYtwo{\Sout\Sin'}}^{(\mu)}\nonumber\\
&\quad\ge D(a_x{(\mu_*)}|x,\CYtwo{\mu=\mu_*})\omega_{\CYtwo{\Sout\Sin'}}^{(\CYtwo{\mu=\mu_*})}=\omega_{\CYtwo{\Sout\Sin'}}^{(\mu_*)}\neq0.
\end{align}
Since $\omega_{\CYtwo{\Sout\Sin'}}^{(\mu_*)}\ge0$ and it is not vanishing, we conclude that ${\rm tr}\left(\omega_{\CYtwo{\Sout\Sin'}}^{(\mu_*)}\right)>0$.
This means that the maximisation in Eq.~\eqref{Eq:SDP} in the main text with $a=a_0(\mu_*)$ and $b=a_1(\mu_*)$ must output a strictly positive value (as $P=\omega_{\CYtwo{\Sout\Sin'}}^{(\mu_*)}$ is already a feasible solution). By applying Result~\ref{Result:Characterising complementarity} in the main text, we conclude that $\mathbfcal{E}$ must 
\CY{have non-vanishing overlap}
(i.e., $\mathbfcal{E}$ is not complementary).
The proof is thus completed.
\end{proof}
\subsection{Supplemental Material III: Proof of Result~\ref{Result:no-go}}
\CYtwo{Since now, as long as $\mathbfcal{E}$ is rank-one projective, it is by definition a pair of instruments with the same input and output spaces, which is denoted by $\Sin=\Sout=S$ with dimension $d<\infty$.}
Before the proof, we need the following lemma:

\begin{lemma}\label{lemma:rank-one projective}
Let $\mathbfcal{E}$ be a pair of \CYtwo{rank-one projective} instruments \CYtwo{acting on $S$} and $\{\ket{a}\}_a$ be a given computational basis \CYtwo{of $S$}.
\CYtwo{Then,} there is a unitary channel $\mathcal{U}(\cdot) \coloneqq U(\cdot)U^\dagger$ with some unitary operator $U$ such that 
\begin{align}\label{Eq: clean form 01}
\mathcal{U}\circ\mathcal{E}_{a|0}\circ\mathcal{U}^\dagger(\cdot) = \proj{a}(\cdot)\proj{a}\quad\forall\,a,
\end{align}
and there is an orthonormal basis $\{\ket{\phi_a}\}_a$ achieving
\begin{align}\label{Eq: clean form 02}
\mathcal{U}\circ\mathcal{E}_{a|1}\circ\mathcal{U}^\dagger(\cdot) = \proj{\phi_a}(\cdot)\proj{\phi_a}\quad\forall\,a.
\end{align}
Finally, if $\mathbfcal{E}$ is also complementary, then 
\begin{align}\label{Eq: complementary rank-one projective condition}
|\braket{a}{\phi_b}|<1\quad \forall\,a,b.
\end{align}
\end{lemma}
\begin{proof}
By definition [i.e., Eq.~\eqref{Eq: Def rank-one projective} in the main text], if $\mathbfcal{E}$ is rank-one projective, then we can write
\begin{align}
\mathcal{E}_{a|x}(\cdot) = \proj{\psi_a^{(x)}}(\cdot)\proj{\psi_a^{(x)}}\quad\forall\,a,x, 
\end{align}
and for each $x=0,1$ the set $\{\ket{\psi_a^{(x)}}\}_a$ is an orthonormal basis.
With the given computational basis $\{\ket{a}\}_a$, consider the unitary operator 
\begin{align}
U\coloneqq\sum_n\ket{n}\bra{\psi_n^{(0)}}.
\end{align}
The corresponding unitary channel $\mathcal{U}(\cdot)\coloneqq U(\cdot)U^\dagger$ achieves
\begin{align}
\mathcal{U}\circ\mathcal{E}_{a|0}\circ\mathcal{U}^\dagger(\cdot) = \proj{a}(\cdot)\proj{a}\quad\forall\,a.
\end{align}
For $x=1$, define the state 
\begin{align}
\ket{\phi_a}\coloneqq U\ket{\psi_a^{(1)}}\quad\forall\,a.
\end{align}
Then, $\{\ket{\phi_a}\}_a$ is by construction an orthonormal basis and
\begin{align}
\mathcal{U}\circ\mathcal{E}_{a|1}\circ\mathcal{U}^\dagger(\cdot) = \proj{\phi_a}(\cdot)\proj{\phi_a}\quad\forall\,a.
\end{align}

Finally, to see Eq.~\eqref{Eq: complementary rank-one projective condition}, suppose the opposite; namely, suppose $\mathbfcal{E}$ is also complementary but $|\braket{a}{\phi_b}|=1$ for some $a,b$.
This means $\ket{a}=c\ket{\phi_b}$ for some complex number $c$ with $|c|=1$, and, consequently, $\proj{a}=\proj{\phi_b}$.
Then, for this pair of outcome indices, we have that
\begin{align}
&\left[(\mathcal{U}\circ\mathcal{E}_{a|0}\circ\mathcal{U}^\dagger)_S\otimes\mathcal{I}_{S'}\right](\Phi^+_{SS'})=\proj{aa}_{SS'}/d\nonumber\\
&\quad=\left[(\mathcal{U}\circ\mathcal{E}_{b|1}\circ\mathcal{U}^\dagger)_S\otimes\mathcal{I}_{S'}\right](\Phi^+_{SS'}),
\end{align}
where $d$ is the dimension of $S$ and $S'$.
Using the property that $(U_S^\dagger\otimes\id_{S'})\ket{\Phi^+}_{SS'} = (\id_S\otimes U^{\dagger,t}_{S'})\ket{\Phi^+}_{SS'}$ and write $\mathcal{U}^{t}(\cdot)\coloneqq U^{t}(\cdot)U^{\dagger,t}$, we have
\begin{align}
\left(\mathcal{E}_{a|0}\otimes\mathcal{I}_{S'}\right)(\Phi^+_{SS'})&=\left(\mathcal{U}^\dagger_S\otimes\mathcal{U}^{t}_{S'}\right)(\proj{aa}_{SS'})/d\nonumber\\
&=\left(\mathcal{E}_{b|1}\otimes\mathcal{I}_{S'}\right)(\Phi^+_{SS'}).
\end{align}
By setting $P_{SS'}=\left(\mathcal{U}^\dagger_S\otimes\mathcal{U}^{t}_{S'}\right)(\proj{aa}_{SS'})/d$, which is non-vanishing and positive semi-definite, Lemma~\ref{lemma: Choi cahracterisation common Q part} implies that $\mathbfcal{E}$ is {\em not} complementary (as it has non-vanishing overlap), resulting in a contradiction.
This shows that $|\braket{a}{\phi_b}|<1$ for {\em every} $a,b$ if $\mathbfcal{E}$ is complementary.
The proof is completed. 
\end{proof}

\CYtwo{Now, using Eqs.~\eqref{Eq: Choi common Q inf def001} and~\eqref{Eq: Choi common Q inf def002} with $\Sin=\Sout=S$, we have 
\begin{fact}\label{fact:unitary equi}
A pair of instruments $\mathbfcal{E}$ is complementary if and only if \mbox{$\{\mathcal{U}\circ\mathcal{E}_{a|x}\circ\mathcal{U}^\dagger\}_{a,x}$} is complementary for every unitary channel $\mathcal{U}$.
\end{fact}}
Hence, Lemma~\ref{lemma:rank-one projective} \CYtwo{and Fact~\ref{fact:unitary equi} jointly imply} that if the given pair of instruments is rank-one projective, then, without loss of generality, we can always assume a relatively clean form for it [namely, Eqs.~\eqref{Eq: clean form 01} and~\eqref{Eq: clean form 02}].
Now, we prove Result~\ref{Result:no-go}.

\subsubsection{Proof of Result~\ref{Result:no-go}}
\begin{proof}
Suppose the opposite: $\mathbfcal{E}$ is complementary, and it can be expressed by Eq.~\eqref{Eq:post-processing} in the main text.
This means there are channels $\{\mathcal{M}_x\}_x$, conditional probabilities $\{P(a|x,\lambda)\}_{a,x,\lambda}$, and a single instrument $\{\mathcal{G}_\lambda\}_\lambda$ achieving
$\mathcal{E}_{a|x} = \sum_\lambda P(a|x,\lambda)\mathcal{M}_x\circ\mathcal{G}_\lambda$ $\forall\,a,x.$
By acting on half of the maximally entangled state $\ket{\Phi^+}_{SS'}$ [Eq.~\eqref{Eq: max ent state}], we obtain (note that \CYtwo{both $\mathbfcal{E}$'s input and output spaces are the $d$-dimensional system $S$}) 
\begin{align}\label{Eq:post-computation001}
&\left(\mathcal{E}_{a|x}\otimes\mathcal{I}_{S'}\right)(\Phi^+_{SS'}) = (\mathcal{M}_x\otimes\mathcal{I}_{S'})\left(\tau_{a|x}^{\rm LHS}\right)\nonumber\\
&\quad= \sum_\mu D(a|x,\mu)(\mathcal{M}_x\otimes\mathcal{I}_{S'})\left(\omega_{SS'}^{(\mu)}\right),
\end{align}
where $\tau_{a|x}^{\rm LHS}$'s are some LHS model given by Eq.~\eqref{Eq:clasiscal model}, and we have 
used Eq.~\eqref{Eq:post-mu} here.
Since $\mathbfcal{E}$ is complementary and rank-one projective, Lemma~\ref{lemma:rank-one projective} \CYtwo{and Fact~\ref{fact:unitary equi} imply that,} without loss of generality, we can assume $\mathbfcal{E}$ takes the following form
\begin{align}\label{Eq:rank-one01}
&\mathcal{E}_{a|0}(\cdot) = \proj{a}(\cdot)\proj{a}\quad\forall\,a;\\\label{Eq:rank-one02}
&\mathcal{E}_{a|1}(\cdot) = \proj{\phi_a}(\cdot)\proj{\phi_a}\quad\forall\,a,
\end{align}
where $\{\ket{a}\}_a$ is the computational basis, and $\{\ket{\phi_a}\}_a$ is an othonormal basis satisfying
$
|\braket{a}{\phi_b}|<1$
$\forall\,a,b.
$
Then, setting $x=0$ in Eq.~\eqref{Eq:post-computation001} and taking the partial trace ${\rm tr}_{S}$ gives
\begin{align}
\proj{a}_{S'}/d = \sum_\mu D(a|0,\mu){\rm tr}_{S}\left(\omega_{SS'}^{(\mu)}\right)\quad\forall\,a.
\end{align}
For every $\mu$, there is \CYtwo{exactly one} value of $a$, \CYtwo{denoted by $a_0(\mu)$ as in Eq.~\eqref{Eq: each mu has at least one a index},} achieving $D(\CYtwo{a=a_0(\mu)}|0,\mu)=1$.
Hence,
\begin{align}\label{Eq: a_mu}
\proj{\CYtwo{a_0(\mu)}}_{S'}/d\ge{\rm tr}_{S}\left(\omega_{SS'}^{(\mu)}\right)\quad\forall\,\mu.
\end{align}
This means that, for some $\delta_\mu\ge0$, 
\begin{align}\label{Eq:omega local condition}
{\rm tr}_{S}\left(\omega_{SS'}^{(\mu)}\right) = \delta_\mu\proj{\CYtwo{a_0(\mu)}}_{S'}\quad\forall\,\mu.
\end{align}
Now, setting $x=1$ in Eq.~\eqref{Eq:post-computation001} and taking ${\rm tr}_{S}$ gives
\begin{align}\label{Eq:omega local condition2}
\CYtwo{\left(\proj{\phi_b}_{S'}\right)^t}/d = \sum_\mu D(b|1,\mu){\rm tr}_{S}\left(\omega_{SS'}^{(\mu)}\right)\quad\forall\,b,
\end{align}
where \CYtwo{$b$ denote the outcome index for $x=1$, and $(\cdot)^t$ is the transpose with respect to the computational basis $\{\ket{a}\}_a$}.
Using Eq.~\eqref{Eq:omega local condition}
\CYtwo{and take the transpose $(\cdot)^t$ on both sides of Eq.~\eqref{Eq:omega local condition2}}, we thus conclude that, for every $b$,
\begin{align}\label{Eq: diagonal condition}
\text{$\proj{\phi_b}$ is diagonal in the basis $\{\ket{a}\}_a$.}
\end{align}
For a fixed index $b$, $\ket{\phi_b}$ is a normalised state, meaning that
\begin{align}
1=\braket{\phi_b}{\phi_b} = \CYtwo{\sum_a|\braket{a}{\phi_b}|^2.}
\end{align}
Together with $|\braket{a}{\phi_b}|<1$ $\forall\,a$ [Eq.~\eqref{Eq: complementary rank-one projective condition}], 
we learn that there must be at least two different indices $a\neq a'$ achieving 
\begin{align}\label{Eq:off-diagonal terms}
|\braket{a}{\phi_b}|\neq0\neq|\braket{a'}{\phi_b}|.
\end{align}
This thus implies that $\braket{a}{\phi_b}\braket{\phi_b}{a'}\neq0$; namely, $\proj{\phi_b}$ has at least one non-vanishing off-diagonal term with respect to the computational basis $\{\ket{a}\}_a$, which contradicts with Eq.~\eqref{Eq: diagonal condition}.
This thus completes the proof.
\end{proof}

\subsection{Supplemental Material IV: Proof of Result~\ref{Result:noise robustness}}
\begin{proof}
Suppose the noisy Pauli instruments $\mathbfcal{P}^{(\epsilon)}$ [Eq.~\eqref{Eq:noisy Pauli} in the main text] can be simulated by quantum post-processing; namely, $\mathbfcal{P}^{(\epsilon)}\in{\bf JM}_{\rm post}$.
Then, using Eq.~\eqref{Eq:post-computation001} and the definition of $\mathbfcal{P}^{(\epsilon)}$, we obtain (below, $H\coloneqq \ket{+}\bra{0}+\ket{-}\bra{1}$ is the Hadamard gate; \CYtwo{also, here, both $S$ and $S'$ are qubits})
\begin{align}
&\left(\mathcal{P}_{a|x}^{(\epsilon)}\otimes\mathcal{I}_{S'}\right)(\Phi^+_{SS'})\nonumber\\
&\quad=\frac{1-\epsilon}{2}(H_S\otimes H_{S'})^x\proj{aa}_{SS'}(H_S\otimes H_{S'})^x+\frac{\epsilon\id_{SS'}}{8}\nonumber\\
&\quad= \sum_\mu D(a|x,\mu)(\mathcal{M}_x\otimes\mathcal{I}_{S'})\left(\omega_{SS'}^{(\mu)}\right)\quad\forall\,a,x.
\end{align}
Taking the partial trace over $S$ gives
\begin{align}\label{Eq: local formula S'}
\frac{1-\epsilon}{2}H_{S'}^x\proj{a}_{S'} H_{S'}^x+\frac{\epsilon\id_{S'}}{4}= \sum_\mu D(a|x,\mu)\tr_{S}\left(\omega_{SS'}^{(\mu)}\right)\;\forall\,a,x.
\end{align}
For simplicity, let us write 
\begin{align}
\widetilde{\omega}^{(\mu)}\coloneqq\tr_{S}\left(\omega_{SS'}^{(\mu)}\right)\quad\forall\,\mu
\end{align} 
and drop the subscript since everything will be in the system $S'$ since now.
When $x=0=a$, we have 
\begin{align}
\frac{1-\epsilon}{2}\proj{0}+\frac{\epsilon\id}{4}= \sum_\mu D(0|0,\mu)\widetilde{\omega}^{(\mu)}.
\end{align}
Now, we observe that there are {\em exactly two} deterministic probability distributions that achieve $D(0|0,\mu)=1$ \CYtwo{(i.e., one maps $x=1$ to $a=0$, and another maps $x=1$ to  $a=1$)}.
This means there are exactly two different $\mu$'s, say, $\mu_+$ and $\mu_-$, such that \CYtwo{$D(0|0,\mu_+)=1=D(0|0,\mu_-)$ and, consequently,}
\begin{align}\label{Eq: exactly two mu eq}
\frac{1-\epsilon}{2}\proj{0}+\frac{\epsilon\id}{4}= \widetilde{\omega}^{(\mu_+)}+ \widetilde{\omega}^{(\mu_-)}.
\end{align}
Hence, there must be one of them (and, without loss of generality, suppose it is $\mu_+$) such that $\widetilde{\omega}^{(\mu_+)}$'s greatest eigenvalue is no less than half of the greatest eigenvalue of \mbox{$[(1-\epsilon)/2]\proj{0}+{\epsilon\id/4}$.}
Namely, we have
\begin{align}\label{Eq:eigenvalue bound}
\norm{\widetilde{\omega}^{(\mu_+)}}_\infty\ge\frac{1}{2}\norm{\frac{1-\epsilon}{2}\proj{0}+\frac{\epsilon\id}{4}}_\infty=\frac{2-\epsilon}{8},
\end{align}
where $\norm{P}_\infty$ is the sup norm of the operator $P$, which is $|P|$'s largest eigenvalue.
Suppose $\ket{\omega_+}$ is $\widetilde{\omega}^{(\mu_+)}$'s eigenstate with the above eigenvalue.
That is,
$
\widetilde{\omega}^{(\mu_+)}\ket{\omega_+}=\norm{\widetilde{\omega}^{(\mu_+)}}_\infty\ket{\omega_+}.
$
This also means that 
\begin{align}
\CYtwo{\norm{\widetilde{\omega}^{(\mu_+)}}_\infty\proj{\omega_+}\le\widetilde{\omega}^{(\mu_+)}.}
\end{align}

\CYtwo{Now, from Eq.~\eqref{Eq: exactly two mu eq}, we have}
%
\begin{align}\label{Eq: tilde omega upper bound 0}
\widetilde{\omega}^{(\mu_+)}\le\frac{1-\epsilon}{2}\CYtwo{\proj{0}}+\frac{\epsilon\id}{4}.
\end{align}
%
Also, when $x=1$, \CYtwo{just like Eq.~\eqref{Eq: each mu has at least one a index},} there is a unique value \CYtwo{$a_1(\mu_+)$} such that $D(\CYtwo{a_1(\mu_+)}|1,\mu_+)=1$.
\CYtwo{Substituting this to Eq.~\eqref{Eq: local formula S'} with $x=1$ and $a=a_1(\mu_+)$ gives}
\begin{align}\label{Eq: tilde omega upper bound +}
\widetilde{\omega}^{(\mu_+)}&=\CYtwo{D(\CYtwo{a_1(\mu_+)}|1,\mu_+)\widetilde{\omega}^{(\mu_+)}}\nonumber\\
&\le\CYtwo{\sum_\mu D(a_1(\mu_+)|1,\mu)\widetilde{\omega}^{(\mu)}}\nonumber\\
&=\frac{1-\epsilon}{2}H\proj{\CYtwo{a_1(\mu_+)}}H+\frac{\epsilon\id}{4}.
\end{align}
Without loss of generality, let us assume $\CYtwo{a_1(\mu_+)=0}$ (and hence $H\ket{\CYtwo{a_1(\mu_+)}}=\ket{+}$), as the case of $\CYtwo{a_1(\mu_+)=1}$ can be shown in the same way.
\CYtwo{Combining Eqs.~(\ref{Eq:eigenvalue bound}-\ref{Eq: tilde omega upper bound +}),}
we obtain
\begin{align}\label{Eq:omega_+ 0 upper bound}
\frac{2-\epsilon}{8}\proj{\omega_+}&\le\CYtwo{\norm{\widetilde{\omega}^{(\mu_+)}}_\infty\proj{\omega_+}\le\widetilde{\omega}^{(\mu_+)}}\nonumber\\
&\le\frac{1-\epsilon}{2}\proj{0}+\frac{\epsilon\id}{4},
\end{align}
and, similarly,
\begin{align}\label{Eq:omega_+ + upper bound}
\frac{2-\epsilon}{8}\proj{\omega_+}&\le\frac{1-\epsilon}{2}\proj{+}+\frac{\epsilon\id}{4}.
\end{align}
\CYthree{Sandwiching Eq.~\eqref{Eq:omega_+ 0 upper bound}  by $\bra{1}$ and $\ket{1}$, we obtain
\begin{align}\label{Eq: proof relation 001}
|\braket{1}{\omega_+}|^2\le\frac{2\epsilon}{2-\epsilon}.
\end{align}
Similarly, sandwiching Eq.~\eqref{Eq:omega_+ + upper bound} by $\bra{-}$ and $\ket{-}$ gives
\begin{align}\label{Eq: proof relation 002}
|\braket{-}{\omega_+}|^2\le\frac{2\epsilon}{2-\epsilon}.
\end{align}}
Using the above relations, a direct computation shows that
\begin{align}
\sqrt{\frac{4\epsilon}{2-\epsilon}}&\ge\sqrt{2}|\braket{-}{\omega_+}|=|\braket{0}{\omega_+}-\braket{1}{\omega_+}|\nonumber\\
&\ge|\braket{0}{\omega_+}|-|\braket{1}{\omega_+}|\ge|\braket{0}{\omega_+}|-\sqrt{\frac{2\epsilon}{2-\epsilon}}.
\end{align}
In other words,
\begin{align}\label{Eq:contradition relation 001}
|\braket{0}{\omega_+}|^2\le\frac{(6+4\sqrt{2})\epsilon}{2-\epsilon}.
\end{align}
Finally, we note that $\ket{\omega_+}$ is a normalised state. Using the normalisation condition and \CYthree{Eq.~\eqref{Eq: proof relation 001}} again, we obtain 
\begin{align}
1=|\braket{0}{\omega_+}|^2+|\braket{1}{\omega_+}|^2\le|\braket{0}{\omega_+}|^2+\frac{2\epsilon}{2-\epsilon}.
\end{align}
Consequently,
\begin{align}\label{Eq:contradition relation 002}
\frac{2-3\epsilon}{2-\epsilon}\le|\braket{0}{\omega_+}|^2.
\end{align}
Importantly, what we have obtained is 
\begin{align}
\text{$\mathbfcal{P}^{(\epsilon)}\in{\bf JM}_{\rm post}$ $\Rightarrow$ Eqs.~\eqref{Eq:contradition relation 001} $\&$~\eqref{Eq:contradition relation 002} hold simultaneously.}
\end{align}
Hence, if Eqs.~\eqref{Eq:contradition relation 001} and~\eqref{Eq:contradition relation 002} {\em cannot} hold simultaneously, we must have $\mathbfcal{P}^{(\epsilon)}\notin{\bf JM}_{\rm post}$.
This is the case if we have \mbox{$(6+4\sqrt{2})\epsilon<2-3\epsilon$},
which is equivalent to
\begin{align}\label{Eq: if condition}
\epsilon<\frac{2}{9+4\sqrt{2}},
\end{align}
which is thus a sufficient condition for $\mathbfcal{P}^{(\epsilon)}\notin{\bf JM}_{\rm post}$.
The proof is thus completed.
\end{proof}

\subsection{Supplemental Material V: Proof of Result~\ref{Result:no-go2}}
\begin{proof}
Suppose the opposite: $\mathbfcal{E}$ is complementary and expressible by Eq.~\eqref{Eq:pre-processing} in the main text.
Namely, we can write
$\mathcal{E}_{a|x} = \sum_\lambda P(a|x,\lambda)\mathcal{G}_\lambda\circ\mathcal{N}_x$ $\forall\,a,x$ for some unital channels $\{\mathcal{N}_x\}_x$, conditional probabilities $\{P(a|x,\lambda)\}_{a,x,\lambda}$, and instrument $\{\mathcal{G}_\lambda\}_\lambda$.
Then, similar to Eq.~\eqref{Eq:post-computation001}, we have (again, \CYtwo{both $\mathbfcal{E}$'s input and output spaces are the $d$-dimensional system $S$})
\begin{align}\label{Eq:computation001}
\left(\mathcal{E}_{a|x}\otimes\mathcal{I}_{S'}\right)(\Phi^+_{SS'}) &= (\mathcal{I}_S\otimes\mathcal{N}_x^t)\left(\tau_{a|x}^{\rm LHS}\right)\nonumber\\
&=\sum_\mu D(a|x,\mu)(\mathcal{I}_S\otimes\mathcal{N}_x^t)\left(\omega_{SS'}^{(\mu)}\right),
\end{align}
where $\tau_{a|x}^{\rm LHS}$'s are LHS model [Eq.~\eqref{Eq:clasiscal model}], and Eq.~\eqref{Eq:post-mu} is used.
By writing \mbox{$\mathcal{N}_x(\cdot) = \sum_i K_i^{(x)}(\cdot)K_i^{(x),\dagger}$} be a Kraus representation of $\mathcal{N}_x$, we define \mbox{$\mathcal{N}_x^t(\cdot)\coloneqq\sum_i K_i^{(x),t}(\cdot)K_i^{(x),\dagger,t}$},
which is again a unital channel since $\mathcal{N}_x$ is so.
\CYtwo{Also, just like in the proof of Result~\ref{Result:no-go}, Lemma~\ref{lemma:rank-one projective} and Fact~\ref{fact:unitary equi} jointly imply that, without loss of generality, we can assume $\mathbfcal{E}$ takes the form of Eqs.~\eqref{Eq:rank-one01} and~\eqref{Eq:rank-one02}, since it is rank-one projective.} 

Now, for $x=0$, after taking the partial trace ${\rm tr}_{S'}$, we have
\begin{align}
\proj{a}_S/d = \sum_\mu D(a|0,\mu){\rm tr}_{S'}\left(\omega_{SS'}^{(\mu)}\right)\quad\forall\,a.
\end{align}
For every $\mu$, there is a value of $a$, 
\CYtwo{denoted by $a_0(\mu)$ as in Eq.~\eqref{Eq: each mu has at least one a index},}
achieving \mbox{$D(\CYtwo{a=a_0(\mu)}|0,\mu)=1$}.
We conclude that
\begin{align}\label{Eq: tr_S' diagonal in computational basis}
\proj{\CYtwo{a_0(\mu)}}/d\ge{\rm tr}_{S'}\left(\omega_{SS'}^{(\mu)}\right)\quad\forall\,\mu.
\end{align}
Namely, each ${\rm tr}_{S'}\left(\omega_{SS'}^{(\mu)}\right)$ is diagonal in the basis $\{\ket{a}\}_a$ [just like Eq.~\eqref{Eq:omega local condition} but now for the system $S$].
Finally, for $x=1$, combining Eqs.~\eqref{Eq:rank-one02} and~\eqref{Eq:computation001} and applying ${\rm tr}_{S'}$ imply that
\begin{align}
\proj{\phi_b}_S = \mathcal{E}_{b|1}(\id_S) = d\sum_\mu D(b|1,\mu){\rm tr}_{S'}\left(\omega_{SS'}^{(\mu)}\right)\quad\forall\,b,
\end{align}
where the right-hand side is diagonal in the basis $\{\ket{a}\}_a$ due to Eq.~\eqref{Eq: tr_S' diagonal in computational basis}.
This thus contradicts with Eq.~\eqref{Eq: complementary rank-one projective condition}, which implies that $\proj{\phi_b}$ {\em must} have some non-vanishing off-diagonal terms in the basis $\{\ket{a}\}_a$ [see the argument around Eq.~\eqref{Eq:off-diagonal terms}].
Hence, complementary $\mathbfcal{E}$ cannot be expressed by Eq.~\eqref{Eq:pre-processing} in the main text.

Finally, to prove that \mbox{$0\le\epsilon<2/\left(9+4\sqrt{2}\right)$} implies \mbox{$\mathbfcal{P}^{(\epsilon)}\notin{\bf JM}_{\rm pre}^{\rm unital}$},
one can substitute $\mathbfcal{P}^{(\epsilon)}$ into Eq.~\eqref{Eq:computation001} and trace out the system $S'$ (instead of $S$). This gives
\begin{align}\label{Eq: local formula S}
\frac{1-\epsilon}{2}H_{S}^x\proj{a}_{S} H_{S}^x+\frac{\epsilon\id_{S}}{4}= \sum_\mu D(a|x,\mu)\tr_{S'}\left(\omega_{SS'}^{(\mu)}\right)\;\;\forall\,a,x,
\end{align}
which is mathematically identical to Eq.~\eqref{Eq: local formula S'} except it is in the system $S$ now.
This thus means the same proof for Result~\ref{Result:noise robustness} in Supplemental Material V can go through to show that Eq.~\eqref{Eq: if condition} is a sufficient condition of $\mathbfcal{P}^{(\epsilon)}\notin{\bf JM}_{\rm pre}^{\rm unital}$. 
\end{proof}

\clearpage
\subsection{Supplemental Material VI: Dual SDP and Proof of Result~\ref{Result:operational meaning}}


Before proving Result~\ref{Result:operational meaning}, we first solve the dual SDP of Eq.~\eqref{Eq:SDP} in the main text.
We note that \CYtwo{this SDP}
is mathematically equivalent to the one studied in Ref.~\cite{Hsieh-IP} [specifically, Eq.~(3) therein].
Hence, using Eq.~(19) and Lemma 8 in Ref.~\cite{Hsieh-IP}, we conclude that \CYtwo{[recall from Eq.~\eqref{Eq:Choi operator diff input output} that $\Choi_{a|x}^{(\mathbfcal{E})}\coloneqq(\mathcal{E}_{a|x}\otimes\mathcal{I}_{\Sin'})(\Phi^+_{\Sin\Sin'})$ $\forall\,a,x$]}
\CYtwo{
\begin{lemma}{\em\cite{Hsieh-IP}}\label{Lemma:Hsieh-IP}
Consider a pair of $\Sin\to\Sout$ instruments $\mathbfcal{E}$.
Then, there is a parameter $0\le\eta_*<1$ such that, for every outcome index pair $(a,b)$ and \mbox{$\eta_*\le\eta<1$}, we have 
\begin{equation}\label{Eq:dualSDP}
\begin{split}
C_{(a,b)}(\mathbfcal{E}) = \frac{1}{1-\eta}\min_{
R,L
}\quad&\tr\left(R\Choi_{a|0}^{\CY{(\mathbfcal{E})}}\right)+\tr\left(L\Choi_{b|1}^{\CY{(\mathbfcal{E})}}\right)\\
{\rm s.t.}\quad&R\ge0,\,L\ge0,\\
&(1-\eta)\id_{\Sout\Sin'}\le R+L\le\id_{\Sout\Sin'},
\end{split}
\end{equation}
which is minimising over operators $R$ and $L$ acting on the bipartite system $\Sout\Sin'$.
\end{lemma}
}
Now, recall from the main text that 
\begin{align}
\CY{{\bm\pi}(\mathbfcal{E})}\coloneqq\left\{\mathcal{E}_{\pi_x(a)|x}\right\}_{a,x}
\end{align} 
is a new instrument pair
with two given permutations \mbox{${\bm\pi}=\{\pi_0,\pi_1\}$} (one for each $x$ value) used to relabel the originally given instrument pair $\mathbfcal{E}$'s classical outcomes. 
Then, using Eq.~\eqref{Eq:dualSDP}, one can show the following lemma: 
\begin{lemma}\label{lemma: complementarity permutation}
A pair of \CYtwo{$\Sin\to\Sout$} instruments $\mathbfcal{E}$ is complementary if and only if 
\begin{align}
C_{(a,a)}\big(\CY{{\bm\pi}(\mathbfcal{E})}\big)=0\quad\forall\,a,{\bm\pi}.
\end{align} 
\end{lemma}
\begin{proof}
({\em Proof of the direction ``$\Leftarrow$''})
Suppose that \mbox{$C_{(a,a)}\big(\CY{{\bm\pi}(\mathbfcal{E})}\big)=0$} for every $a$ and permutation pair ${\bm\pi}$.
For each outcome index pair $(a,b)$, one can consider any ${\bm\pi}$ with $\pi_0(a)=a$ and $\pi_1(a)=b$ to achieve \CYtwo{$\mathcal{E}_{a|0}=\mathcal{E}_{\pi_0(a)|0}$ and $\mathcal{E}_{b|1}=\mathcal{E}_{\pi_1(a)|1}$, meaning that}
\begin{align}
\CYtwo{J_{a|0}^{(\mathbfcal{E})}=J_{a|0}^{({\bm\pi}(\mathbfcal{E}))}\quad\&\quad J_{b|1}^{(\mathbfcal{E})}=J_{a|1}^{({\bm\pi}(\mathbfcal{E}))},}
\end{align}
\CYtwo{where 
\begin{align}\label{Eq:Choi permutation}
\Choi_{a|x}^{({\bm\pi}(\mathbfcal{E}))}=(\mathcal{E}_{\pi_x(a)|x}\otimes\mathcal{I}_{\Sin'})(\Phi^+_{\Sin\Sin'})\quad\forall\,a,x
\end{align} 
are the Choi operators of the instrument pair $\CY{{\bm\pi}(\mathbfcal{E})}$.}
Using Eq.~\eqref{Eq:SDP} in the main text, this thus implies that
\begin{align}
C_{(a,b)}(\mathbfcal{E})=C_{(a,a)}\big(\CY{{\bm\pi}(\mathbfcal{E})}\big)=0.
\end{align}
Using Result~\ref{Result:Characterising complementarity}, this thus means $\max_{a,b}C_{(a,b)}(\mathbfcal{E})=0$ and hence $\mathbfcal{E}$ is complementary.

({\em Proof of the direction ``$\Rightarrow$''})
Now, suppose $\mathbfcal{E}$ is complementary.
Result~\ref{Result:Characterising complementarity} implies 
\begin{align}
C_{(a,b)}(\mathbfcal{E})=0\quad\forall\,a,b.
\end{align}
Then, for every $a$ and permutation ${\bm\pi}$, by setting $a'=\pi_0(a)$ and $b'=\pi_1(a)$, we have 
\begin{align}
C_{(a,a)}\big(\CY{{\bm\pi}(\mathbfcal{E})}\big)=C_{(a',b')}(\mathbfcal{E})=0,
\end{align}
which thus concludes the proof.
\end{proof}

Combining Lemma~\ref{Lemma:Hsieh-IP} and Lemma~\ref{lemma: complementarity permutation}, 
we can now prove Result~\ref{Result:operational meaning} as follows.

\subsubsection{Proof of Result~\ref{Result:operational meaning}}
\begin{proof}
({\em Proof of the direction ``$\Leftarrow$''})
Suppose there is some $0\le\eta<1$ such that $P_{\rm error}^{(\eta)}\big(\CY{{\bm\pi}(\mathbfcal{E})}\big)=0$ for all permutation pairs ${\bm\pi}$.
Consider a given permutation pair ${\bm\pi}$ and an index $a$.
Then, from Eq.~\eqref{Eq:ave error probability_main} in the main text, there exists a POVM $\{Q_0^{(a),({\bm\pi})},Q_1^{(a),({\bm\pi})},Q_\emptyset^{(a),({\bm\pi})}\}$ with \mbox{$Q_\emptyset^{(a),({\bm\pi})}\le\eta\id_{\CYtwo{\Sout\Sin'}}$} achieving \CYtwo{(note that $\Sin$ here is the system $S$ in the main text)}
\begin{align}
0=\CY{p_0}{\rm tr}\left(Q_1^{(a),({\bm\pi})}\Choi_{a|0}^{({\bm\pi}(\mathbfcal{E}))}\right)+\CY{p_1}{\rm tr}\left(Q_0^{(a),({\bm\pi})}\Choi_{a|1}^{({\bm\pi}(\mathbfcal{E}))}\right).
\end{align}
Since $\CY{p_x>0}$ as assumed in the task, we obtain
\begin{align}\label{Eq:computation for dual SDP proof}
{\rm tr}\left(Q_1^{(a),({\bm\pi})}\Choi_{a|0}^{({\bm\pi}(\mathbfcal{E}))}\right)=0={\rm tr}\left(Q_0^{(a),({\bm\pi})}\Choi_{a|1}^{({\bm\pi}(\mathbfcal{E}))}\right).
\end{align}
Now, $Q_\emptyset^{(a),({\bm\pi})}\le\eta\id_{\CYtwo{\Sout\Sin'}}$ implies 
\begin{align}
(1-\eta)\id_{\CYtwo{\Sout\Sin'}}&\le\CYtwo{\id_{\CYtwo{\Sout\Sin'}}-Q_\emptyset^{(a),({\bm\pi})}}\nonumber\\
&=Q_0^{(a),({\bm\pi})}+Q_1^{(a),({\bm\pi})}\le\id_{\CYtwo{\Sout\Sin'}}.
\end{align}
This means $Q_0^{(a),({\bm\pi})}$ and $Q_1^{(a),({\bm\pi})}$ form a feasible solution to the minimisation in the dual SDP in Eq.~\eqref{Eq:dualSDP} when we put $\Choi_{a|0}^{({\bm\pi}(\mathbfcal{E}))}$ and $\Choi_{a|1}^{({\bm\pi}(\mathbfcal{E}))}$ in its objective function.
Together with Eq.~\eqref{Eq:computation for dual SDP proof}, we conclude that
\begin{align}
C_{(a,a)}\big(\CY{{\bm\pi}(\mathbfcal{E})}\big)=0.
\end{align}
Since the above argument works for {\em every} ${\bm\pi}$ and $a$, we thus conclude that $\mathbfcal{E}$ is indeed complementary by using Lemma~\ref{lemma: complementarity permutation}.

({\em Proof of the direction ``$\Rightarrow$''})
Now, to show the opposite direction, suppose $\mathbfcal{E}$ is complementary.
Then, \CYtwo{using Lemma~\ref{Lemma:Hsieh-IP},}
there is a parameter $0\le\eta_*<1$ such that, for {\em every} pair $(a,b)$, there are some $R^{(a,b)}\ge0$ and $L^{(a,b)}\ge0$ with 
\begin{align}\label{Eq: 1-eta condition for R L}
(1-\eta_*)\id_{\CYtwo{\Sout\Sin'}}\le R^{(a,b)}+L^{(a,b)}\le\id_{\CYtwo{\Sout\Sin'}}
\end{align} 
that can achieve
\begin{align}\label{Eq: Result 6 proof 001}
\tr\left(R^{(a,b)}\Choi_{a|0}^{\CY{(\mathbfcal{E})}}\right)=0=\tr\left(L^{(a,b)}\Choi_{b|1}^{\CY{(\mathbfcal{E})}}\right).
\end{align}

Now, consider a given permutation pair ${\bm\pi}=(\pi_0,\pi_1)$ and a given index $a$.
By letting 
\begin{align}
a'=\pi_0(a)\quad\&\quad b'=\pi_1(a)
\end{align} 
and use Eqs.~\eqref{Eq:Choi operator diff input output} and~\eqref{Eq:Choi permutation}, we have 
\begin{align}
\Choi_{a'|0}^{\CY{(\mathbfcal{E})}}=\Choi_{a|0}^{({\bm\pi}(\mathbfcal{E}))}\quad\&\quad\Choi_{b'|1}^{\CY{(\mathbfcal{E})}}=\Choi_{a|1}^{({\bm\pi}(\mathbfcal{E}))}.
\end{align}
Using Eq.~\eqref{Eq: Result 6 proof 001} for the pair $(a',b')$, we obtain
\begin{align}\label{Eq: Result 6 proof 002}
\tr\left(R^{(a',b')}\Choi_{a|0}^{({\bm\pi}(\mathbfcal{E}))}\right)=0=\tr\left(L^{(a',b')}\Choi_{a|1}^{({\bm\pi}(\mathbfcal{E}))}\right).
\end{align}
Hence, for the fixed ${\bm\pi}$ and $a$, we can define the following three-outcome POVM:
\begin{align}
&Q_0^{(a),({\bm\pi})}\coloneqq L^{(a',b')},\\ 
&Q_1^{(a),({\bm\pi})}\coloneqq R^{(a',b')},\\
&Q_\emptyset^{(a),({\bm\pi})}\coloneqq\id_{\CYtwo{\Sout\Sin'}}-R^{(a',b')}-L^{(a',b')}.
\end{align}
Since $(1-\eta_*)\id_{\CYtwo{\Sout\Sin'}}\le R^{(a,b)}+L^{(a,b)}\le\id_{\CYtwo{\Sout\Sin'}}$ [Eq.~\eqref{Eq: 1-eta condition for R L}], one can directly check that the above is a valid POVM with 
\begin{align}
Q_\emptyset^{(a),({\bm\pi})}\le\eta_*\id_{\CYtwo{\Sout\Sin'}},
\end{align}
meaning that this POVM is indeed allowed in the task with the parameter $\eta_*$.  
Consequently, we have
\begin{align}
0=\CY{p_0}{\rm tr}\left(Q_1^{(a),({\bm\pi})}\Choi^{({\bm\pi}(\mathbfcal{E}))}_{a|0}\right)+\CY{p_1}{\rm tr}\left(Q_0^{(a),({\bm\pi})}\Choi^{({\bm\pi}(\mathbfcal{E}))}_{a|1}\right).
\end{align}
Note that the above argument works for every $a$ with the given ${\bm\pi}$.
Together with Eq.~\eqref{Eq:ave error probability_main} in the main text, we conclude that, for the parameter $0\le\eta_*<1$, we have
\begin{align}
P_{\rm error}^{(\eta_*)}\big(\CY{{\bm\pi}(\mathbfcal{E})}\big)=0.
\end{align}
Finally, one can further note that the whole argument works for {\em every} ${\bm\pi}$.
This thus completes the proof.
\end{proof}

\bibliography{Ref.bib}

\begin{thebibliography}{61}
\bibitem [{Foo()}]{FootnoteFig}%
  \BibitemOpen
  \href@noop {} {}\bibinfo {note} {\CYtwo{For any filter $\mathcal{L}(\cdot)=\sum_iK_i(\cdot)K_i^\dagger$ with Kraus operators $K_i$'s satisfying $\Gamma\coloneqq\sum_iK_i^\dagger K_i\le\id$, we can define another filter $\widetilde{\mathcal{L}}(\cdot)\coloneqq\sqrt{\id-\Gamma}(\cdot)\sqrt{\id-\Gamma}^\dagger$. Then, $\mathcal{L}+\widetilde{\mathcal{L}}$ is a channel, meaning that $\{\mathcal{L},\widetilde{\mathcal{L}}\}$ is an instrument. Hence, a filter can always be viewed as part of an instrument.}}\BibitemShut {Stop}%
\end{thebibliography}

\end{document}